\documentclass[]{article}
\usepackage{amsfonts}
\usepackage{amsthm}
\usepackage{youngtab}
\usepackage{amsmath,amscd}
\usepackage{graphicx}
\usepackage{caption}
\usepackage{subcaption}
\usepackage{listings} 
\usepackage{amssymb}
\usepackage{mathtools}
\usepackage[linesnumbered,vlined,ruled,commentsnumbered]{algorithm2e}

\makeatletter
\def\BState{\State\hskip-\ALG@thistlm}
\makeatother

\newtheorem{theorem}{Theorem}[section]

\theoremstyle{definition}

\newtheorem{definition}{Definition}[section]
\newtheorem{claim}{Claim}[section]
\newtheorem{lemma}{Lemma}[section]

\newtheorem{corollary}{Corollary}[section]
\newtheorem{example}{Example}[section]

\newtheorem{proposition}{Proposition}

\title{Sensitivity Analysis of Monotone Submodular Function Maximization}
\author{
	Conor McMeel\\
	Imperial College, London\\
	\texttt{c.mcmeel18@imperial.ac.uk}
	\and
	Yuichi Yoshida \footnote{This author was supported by JST, PRESTO Grant Number JPMJPR192B, Japan.}\\
	National Institute of Informatics\\
	\texttt{yyoshida@nii.ac.jp}
}

\usepackage[utf8]{inputenc}
\usepackage[T1]{fontenc}    % use 8-bit T1 fonts
\usepackage{hyperref}       % hyperlinks
\usepackage{url}            % simple URL typesetting
\usepackage{booktabs}       % professional-quality tables
\usepackage{amsfonts}       % blackboard math symbols
\usepackage{amsmath}
\usepackage{amsthm}
\usepackage{color}
\usepackage[noend]{algpseudocode}
\usepackage[linesnumbered,ruled,vlined]{algorithm2e}
\SetKwProg{Procedure}{Procedure}{}{}
\SetKwInput{Input}{input}
\usepackage{graphicx}
\usepackage{bm}
\allowdisplaybreaks

\begin{document}
	
\maketitle

\begin{abstract}
	We study the recently introduced idea of worst-case sensitivity for monotone submodular maximization with cardinality constraint $k$, which captures the degree to which the output argument changes on deletion of an element in the input. We find that for large classes of algorithms that non-trivial sensitivity of $o(k)$ is not possible, even with bounded curvature, and that these results also hold in the distributed framework. However, we also show that in the regime $k = \Omega(n)$ that we can obtain $O(1)$ sensitivity for sufficiently low curvature.
\end{abstract}
%\ynote{Present our results in terms of worst-case sensitivity and then derive results on average sensitivity as corollaries.}
%\cnote{TO DO: Add proportional greedy proof - Update: DONE}

\section{Introduction}
Monotone submodular function maximization is an important problem in optimization that captures many problems of interest including Max-Cut~\cite{goemans1995improved}, facility location~\cite{ageev19990}, and generalized assignment~\cite{feige2006approximation}. It also has many practical applications such as optimal budget allocation~\cite{soma2014optimal}, sensor placement~\cite{krause2007near} and text summarization~\cite{lin2011class}. It has been studied extensively, and many approximation algorithms with provable guarantees have been derived for a variety of constraints~\cite{calinescu2011maximizing,iyer2013submodular}, as well as a variety of generalizations~\cite{ward2016maximizing,soma2018maximizing}.
%, and various generalisations including $k$-submodular functions~\cite{huber2012towards} and submodular functions over ~\cite{bian2016guaranteed}.

Most works studying monotone submodular function maximization are concerned solely with the approximation guarantee and time complexity of the algorithm in question, assuming that we have full access to the submodular function.
In practical applications, however, it may be naive to assume that we have access to the entire ground set and the ground set does not change over time.
In these cases, it would be desirable to obtain an algorithm whose optimal solution does not change greatly as the ground set receives or loses elements.
We provide two examples below.
\begin{itemize}
	\itemsep=0pt
	\item In a recommendation system~\cite{El-Arini2009}, it is not the value of the objective function that is important to the user, but rather the solution, that is, the items recommended.
	If for example we are recommending items of media, it would be desirable for the user that even after our library of available media changes, the items recommended do not change too much.
	\item In sensor placement~\cite{krause2007near}, if the locations that we can install sensors change over time, it may be desirable to use an algorithm where the solution, that is, the locations we install sensors, does not change too much if installing and uninstalling sensors is a cost or time-intensive process.
\end{itemize}
To address the issue above, in this work, we will be concerned with the \emph{(worst-case) sensitivity} of algorithms for maximizing submodular functions, which is informally how much the output of a maximization algorithm changes upon deleting an element of the ground set (see Section~\ref{sec:pre} for details).

Furthermore, in typical applications for submodular maximization, like those listed above, we typically deal with such large amounts of data that we cannot store it all on one machine. Specifically with reference to our earlier examples, large deployments of sensors mean we now have many physical measurements to store. Additionally, the popularization and increasing user-activity of services that recommend media result in an ever-increasing amount of data. As this data can be high-dimensional, it often must be processed in a distributed fashion. As sensitivity has many applied interests, we also look at distributed algorithms.
% A similar notion, \emph{(average-case) sensitivity}, was initially studied in~\cite{varma2019average}.

\subsection{Our Contributions}
In this work, we analyze the sensitivity of algorithms for maximizing monotone submodular functions under a cardinality constraint.
Note that we have a trivial upper bound of $O(k)$ for any algorithm that outputs a set of size $k$.

We first study algorithms in a general sequential setting, where one element is selected to the output set after another. We first show that algorithms in a class we call \textit{proportional} (meaning elements with large marginals must be chosen with high probability) attain $\Omega(k)$ sensitivity for some functions, which matches the trivial upper bound of $O(k)$. We show the same for the randomized greedy algorithm studied in~\cite{buchbinder2014submodular}.

To get around the lower bound of $\Omega(k)$ for the general case, it is natural to consider monotone submodular functions with bounded curvature, where \emph{curvature}~\cite{Conforti1984} being close to zero implies that the function is close to being modular. However our result is also negative---we show that for algorithms where the probability of selection at each step depends only on the ordinal position of a marginal (and not the marginal itself), that some functions have $\Omega(k)$ sensitivity. Throughout we only have light, natural assumptions on the algorithms. For example, the result applies to the standard greedy algorithm as well as the randomized greedy algorithm of~\cite{buchbinder2014submodular}.

%Our first result is negative: We derive a tight lower bound of $\Omega(k)$ for any ``natural'' sequential algorithm.
%We note that the (deterministic) greedy algorithm~\cite{Nemhauser:1978dm,Nemhauser1978}, which selects the element with the maximum marginal gain $k$ times, is a natural sequential algorithm
%\ynote{To complete the story, it would be nice if we can show any sequential algorithm with constant approximation ratio is natural.}

We then show that if we sacrifice the constant-factor approximation for all $k$ and $n$, then we can obtain lower sensitivity for the proportional greedy algorithm, which selects an element with probability proportional to its marginal gain $k$ times. We show this algorithm has a constant factor approximation when $k = \Omega(n)$, and that its sensitivity is $O((1-\sqrt{1-c})^2/c \cdot k)$.
Note that this bound converges to zero as $c \to 0$, which corresponds to the modular case, and converges to $k$ as $c \to 1$, which corresponds to the general case.
We complement our upper bound by showing that the proportional greedy algorithm has sensitivity $\Omega((1-\sqrt{1-c})^2/c \cdot k - \varepsilon)$ for any $\varepsilon > 0$.
%\ynote{It would be great if we can show general lower/upper bounds without specifying algorithms.}

% We first show that for the deterministic greedy algorithm, there is no bounded $c > 0$ such that the sensitivity is non-trivial. \ynote{Does this mean for any $c>0$ the sensitivity of the deterministic greedy is $\Omega(k)$?}
% For randomized greedy, we derive an upper bound for the average sensitivity that gives us a maximum curvature for low average sensitivity. We also give an explicit function that gives us a lower bound. This is done both for average sensitivity and worst case sensitivity.
% These results for the average sensitivity case are also given for the worst case scenario, where we assume an adversary has free choice to delete one of the elements of the set.

%\cnote{Added to introduction}
%\ynote{We need motivation as to why we study distributed algorithms. Just because it provides another well studied class of algorithms?}

We finally move on to studying distributed algorithms. As many of our hardness results rely on one element changing the ordering of marginals greatly, it could be hoped that splitting the ground set among a large number of machines would enable us to obtain a non-trivial upper bound for all $k, n$, as that one element can't affect too much of the ground set. However, we show this is not the case. Specifically, we show that for a framework given in~\cite{barbosa2016new} for creating distributed algorithms from centralized algorithms, that all algorithms we had currently shown to have $\Omega(k)$ sensitivity in the centralized case will also have $\Omega(k)$ sensitivity in the distributed case. We conclude with a brief discussion. %\cnote{Todo: Add references to average theorems}

\subsection{Related Work}
Sensitivity analysis of algorithms for combinatorial optimization problems is relatively a new area.
Varma and Yoshida~\cite{varma2019average} introduced the notion of \emph{average sensitivity}, which is how much the output of an algorithm changes upon deleting a random element of the ground set, and designed algorithms with low average sensitivity for various problems such as the minimum problem, the maximum matching problem, and the vertex cover problem.
Peng and Yoshida~\cite{Peng2020} studied spectral clustering and showed that it has low average sensitivity when the input graph has a cluster structure.
Yoshida and Zhou~\cite{yoshida2020sensitivity} showed a $(1-\varepsilon)$-approximation algorithm for the maximum matching problem with sensitivity $O_\epsilon(1)$.
Lattanzi~\emph{et~al.}~\cite{lattanzi2020fully} gave a $(1/2-\varepsilon)$ algorithm with polylogarithmic amortized update time in the fully dynamic setting.
Note that their algorithm does not imply an algorithm with low sensitivity because the output of a dynamic algorithm can depend on the order of elements in the stream.

\subsection{Organization}
We introduce notions used throughout this work in Section~\ref{sec:pre}.
In Section~\ref{sec:sequential}, we show that natural sequential algorithms have sensitivity $\Omega(k)$.
In Section~\ref{sec:curvature}, we discuss sensitivity of algorithms on submodular functions with bounded curvature.
In Section~\ref{sec:distributed}, we discuss distributed algorithms, and show how for a common framework for constructing distributed algorithms from centralized algorithms, our results carry over to that setting.
We conclude this paper in Section~\ref{sec:conclusion}. Comparable results for the average-sensitivity setting, as opposed to our worst-case setting, are found in Appendix~\ref{appendix:AveAll}

\section{Preliminaries}\label{sec:pre}
For a positive integer $n$, let $[n]$ denote the set $\{1,2,\ldots,n\}$.
For subsets $S,T \subseteq E$, $S \bigtriangleup T$ denotes the symmetric difference between $S$ and $T$.
For a subset $S \subseteq E$, and an element $e \in E \setminus S$, we denote by $S+e$ the set $S \cup \{e\}$.
For a set $S \subseteq E$, let $1_S \in {\{0,1\}}^E$ be the indicator vector of $S$.
Let $f\colon 2^E \to \mathbb{R}$ be a set function.
For an element $e \in E$, let $f^{\setminus e}\colon 2^{E \setminus \{e\}}\to \mathbb{R}$ be the function obtained from $f$ by restricting the domain to $E \setminus \{e\}$.
For a set $S \subseteq E$, let $f_S\colon 2^{E \setminus S} \to \mathbb{R}$ be the function defined as $f_S(T) = f(S \cup T)$.
% 	Abusing notations, we sometimes identify a set function $f \colon 2^E \to \mathbb{R}$ with a function $f' \colon {\{0,1\}}^E \to \mathbb{R}$ such that $f(S) = f'(1_S)$ for every $S \subseteq E$.
% 	For a random variable $X$, let $\{X\}$ denote its distribution.
% 	The distribution of $X$ conditioned on an event $A$ is denoted by $\{X \mid A\}$.

\subsection{Submodular Functions}
In this subsection, we introduce some necessary definitions.
Let $f\colon 2^E \to \mathbb{R}$ be a set function.
We say that $f$ is \emph{monotone} if $f(S) \leq f(T)$ for any $S \subseteq T$.
We say that $f$ is \emph{submodular} if $f(S) + f(T) \geq f(S \cup T) + f(S \cap T)$.
It is well known that submodularity is equivalent to the \textit{diminishing return property}, that is, $f_S(e) \geq f_T(e)$ holds for any $S \subseteq T$ and $e \in E \setminus T$.
% Note that we can alternatively represent a submodular function as a function on  by identifying a set $S \subseteq E$ with the characteristic vector $1_S \in {\{0,1\}}^E$ of $S$.
% , where the $0-1$ vector corresponds to whether the element is selected or not. \ynote{not clear.}
A function $f$ is \emph{supermodular} if $-f$ is submodular, and is \emph{modular} if it is both submodular and supermodular, which can be thought of as a set function analogue to linearity.
% 	Throughout this work we will concern ourselves with \textit{monotone} submodular functions, which have $f(A) \geq f(B)$ for any $B \subseteq A$.

\begin{algorithm}[t!]
	\caption{Deterministic and randomized greedy algorithms}\label{alg:greedy}
	\Procedure{\emph{\Call{DeterministicGreedy}{$f,k$}}}{
		\Input{Monotone submodular function $f\colon 2^E \to \mathbb{R}_+$ and an integer $k$.}
		$S \gets \emptyset$\;
		\While{$|S| \leq k$}{
			Let $e \in E \setminus S$ be the element with the maximum marginal gain $f_S(e)$\;
			$S \gets S \cup \{e\}$.
		}
		\textbf{return} $S$\;
	}
	\Procedure{\emph{\Call{RandomizedGreedy}{$f,k$}}}{
		\Input{Monotone submodular function $f\colon 2^E \to \mathbb{R}_+$ and an integer $k$.}
		$S \gets \emptyset$\;
		\While{$|S| \leq k$}{
			Let $M$ be a set of size $k$ maximizing $\sum_{e \in M} f(S \cup \{e\})$\;
			Select the element $e'$ uniformly from $M$\;
			$S \gets S \cup \{e'\}$.
		}
		\textbf{return} $S$\;
	}
	
\end{algorithm}

Throughout, the optimization problem we are interested in is the following monotone submodular function maximization with a cardinality constraint:
\begin{quote}
	Given a monotone submodular function $f\colon 2^E \rightarrow \mathbb{R}$, and an integer $1 \leq k \leq  |E|$, find a subset $S \subseteq E$ which maximizes $f(S)$ subject to $|S| \leq k$.
\end{quote}
Although this problem is NP-hard, as it includes the maximum cut problem, it is known that the deterministic greedy algorithm (\Call{DeterministicGreedy}{} in Algorithm~\ref{alg:greedy}) achieves $(1-1/e)$-approximation~\cite{Nemhauser1978}, and $(1-1/e+\varepsilon)$-approximation is NP-Hard for any $\varepsilon > 0$~\cite{Nemhauser:1978dm}.
We can also use the randomized greedy algorithm (\Call{RandomizedGreedy}{} in Algorithm~\ref{alg:greedy}) to get $(1-1/e)$-approximation~\cite{buchbinder2014submodular}.
% These algorithms will output a set $S \subseteq E$ with $f(S) \geq (1-e^{-1})f(\mathrm{OPT})$.

The \textit{curvature} of a submodular function measures how close a submodular function is to modular, and can be defined as follows:
\begin{definition}[\cite{Conforti1984,vondrak2010submodularity}]
	Let $f\colon 2^E \to \mathbb{R}_+$ be a monotone submodular function. Its \textit{curvature} is defined to be
	\[
	c= 1 - \min_{e \in E}\frac{f_{E \setminus \{e\}}(e)}{f(e)}
	\]
\end{definition}
Note that $c \in [0,1]$ and a curvature of $0$ implies that the function is modular.

\subsection{Sensitivity}
In this section, we formally introduce the notion of sensitivity. We start with deterministic algorithms.
\begin{definition}\label{defn:sensDet}
	Let $\mathcal{A}$ be a deterministic algorithm that takes a function $f\colon 2^E \to \mathbb{R}$ and returns a set $S \subseteq E$.
	The \emph{sensitivity} of $\mathcal{A}$ on $f\colon2^E \to \mathbb{R}$ is
	\[
	\max_{e \in E} |\mathcal{A}(f) \bigtriangleup \mathcal{A}(f^{\setminus  e})|.
	\]
\end{definition}
% We can regard the sensitivity of an algorithm $\mathcal{A}$ as the Lipschitz constant of $\mathcal{A}$ by looking at $\mathcal{A}$ as a function that takes $f$, embe and outputs a subset of $E$.
% In our case, we will consider what happens on deleting an element in our ground set. To be specific, if our algorithm runs on a ground set $E = \{e_1, \ldots, e_n\}$ we want to compute:
% submodular maximisation problems that can be represented as a graph, rather than deleting an edge in the graph, we are deleting a vertex.

To define the sensitivity of randomized algorithms, we need the following notion that measures the distance between distributions over subsets.
\begin{definition}
	For two distributions $\mathcal{D}_1$ and $\mathcal{D}_2$ over $2^E$, the \emph{earth mover's (or Wasserstein) distance} between $\mathcal{D}_1$ and $\mathcal{D}_2$ is
	\[
	d_{\mathrm{EM}}(\mathcal{D}_1,\mathcal{D}_2) := \min_{\mathcal{D}} \mathop{\mathbf{E}}_{(S,S') \sim \mathcal{D}}|S \bigtriangleup S'|,
	\]
	where $\mathcal{D}$ is a distribution over $2^E \times 2^E$ such that the marginal distributions on the first and the second coordinates are equal to $\mathcal{D}_1$ and $\mathcal{D}_2$, respectively.
\end{definition}

% For a randomized algorithm $\mathcal{A}$ that takes a function $f\colon 2^E \to \mathbb{R}$ and returns a set $S \subseteq E$, abusing notations, we denote by $\mathcal{A}(f)$ the distribution of the output of $A$ on $f$.
\begin{definition}\label{defn:sensRand}
	Let $\mathcal{A}$ be a randomized algorithm that takes a function $f\colon 2^E \to \mathbb{R}$ and returns a set $S \subseteq E$.
	The \emph{sensitivity} of $\mathcal{A}$ on $f$ is
	\[
	\max_{e \in E} d_{\mathrm{EM}}(\mathcal{A}(f), \mathcal{A}(f^{\setminus  e})),
	\]
	where we identified $\mathcal{A}(f)$ and $\mathcal{A}(f^{\setminus e})$ with the distributions of $\mathcal{A}(f)$ and $\mathcal{A}(f^{\setminus e})$, respectively.
\end{definition}
We note that this definition matches the previous one for deterministic algorithms.

% The average-case version, where $\max_{e \in E}$ is replaced with $\mathop{\mathbf{E}}_{e \sim E}$, was already studied in~\cite{varma2019average} under the name of \emph{average sensitivity}.
% Although the main focus of this work is worst-case sensitivity, we will discuss in Appendix that our results can be extended to average sensitivity. \ynote{todo}

\section{Sequential Algorithms: A Warm-up}\label{sec:sequential}
\begin{algorithm}[t!]
	\caption{Sequential algorithm for monotone submodular function maximization}\label{alg:sequential}
	\Procedure{\emph{\Call{Sequential}{$f,k,\mathcal{R}$}}}{
		\Input{Monotone submodular function $f\colon 2^E \to \mathbb{R}_+$, an integer $k$, and a decision rule $\mathcal{R}$.}
		$S \gets \emptyset$\;
		\While{$|S| \leq k$}{
			$e \leftarrow \mathcal{R}(f,S)$\;
			$S \gets S \cup \{e\}$.
		}
		\textbf{return} $S$\;
	}
\end{algorithm}

% \begin{algorithm}[t!]
% 	\caption{Randomized algorithm for monotone submodular function maximization}\label{alg:randomized}
% 	\Procedure{\emph{\Call{RandomizedGreedy}{$f,k,\mathcal{R}$}}}{
% 		\Input{Monotone submodular function $f\colon 2^E \to \mathbb{R}_+$, an integer $k$.}
% 		$S \gets \emptyset$\;
% 		\While{$|S| \leq k$}{
% 			Let $M$ be a set of size $k$ maximizing $\sum_{e \in M} f(S \cup \{e\})$\;
% 			Select the element $e'$ uniformly from $M$\;
% 			$S \gets S \cup \{e'\}$.
% 		}
% 		\textbf{return} $S$\;
% 	}
% \end{algorithm}
% In this section, we analyze the sensitivity of general submodular functions.
In this section, we consider sequential algorithms, which add elements to the output one by one.
In Section~\ref{subsec:proportional}, we show that any sequential algorithm in a broad class we call \textit{proportional algorithms} has sensitivity $\Omega(k)$, which matches the trivial upper bound of $O(k)$.
In Section~\ref{subsec:randomized-greedy}, we also show the same lower bound for \Call{RandomizedGreedy}{}~\cite{buchbinder2014submodular}.
% The definition is given in Algorithm~\ref{alg:randomized}.

% \begin{lemma}
% 	For the deterministic greedy algorithm for monotone submodular function maximization, the average sensitivity is at most $\mathcal{O}(k^2/n)$. \label{trivupper}
% \end{lemma}
% \begin{proof}
% 	Our bound comes from assuming that if we delete an element not selected by the greedy algorithm (which occurs with probability $(n-k)/n$), that there will be no difference in output, but if we select an element selected at step $i$, then every subsequent element could be different, giving us:
% 	\[
% 	\beta(f) \leq \sum_{i=1}^{k} \frac{k+1-i}{n} = \mathcal{O}\left(\frac{k^2}{n}\right)
% 	\]
% 	as required.
% \end{proof}
\subsection{Proportional Algorithms}\label{subsec:proportional}
% In this section, we develop the formalism needed to define proportional algorithms, and then show a function that achieves $\Omega(k)$ sensitivity.
First, we formally define a decision rule for selecting one element.
\begin{definition}[Decision Rule]
	A (possibly randomized) procedure is called a \emph{decision rule} if it takes a function $f\colon 2^E \to \mathbb{R}$ and a set $S \subseteq E$ and returns an element $e \in E \setminus S$ solely based on the marginal gains $f_S(e) \; (e \in E \setminus S)$.
\end{definition}

The framework of a sequential algorithm is described in Algorithm~\ref{alg:sequential}.
Starting with an empty set, it keeps choosing an element using the given decision rule and then adding it to the set, until the size of the set reaches $k$.
Noting that the choice of $\mathcal{R}$ defines the behavior of this algorithm, we give a couple of examples to make this framework clearer:
\begin{example}
	If $\mathcal{R}(f,S)$ returns $e \in E \setminus S$ with the maximum marginal gain, then we recover \Call{DeterministicGreedy}{}.
	% 	Let $e$ be the position of the maximal element in $(m_1, \ldots, m_n)$. Then if we choose $g$ as:
	% 	\begin{equation}
	% 		g(m_1, \ldots, m_n) = e_i
	% 	\end{equation}
	Suppose that all marginals are different at each step for simplicity.
	Then if $\mathcal{R}(f,S)$ returns an element $e \in E \setminus S$ with probability $1/k$ if it is among the $k$ highest marginals and $0$ otherwise, then we recover \Call{RandomizedGreedy}{}.
	% \ynote{``recover'' sounds weird because we didn't introduce the randomized greedy yet. Also, I think there's no consensus here. Some randomized greedy I know takes top $k$ elements in terms of marginal gain and add one of them chosen uniformly at random to the output set.}
\end{example}

We note that the choice of $\mathcal{R}$ can result in a very broad class of algorithms.
With this in mind, we limit the choice to one we call \textit{proportional decision rules} for this warm-up section, which formalize the intuition that if a marginal is arbitrarily large, we should choose it with arbitrarily high probability.
\begin{definition}[Proportional Decision Rule]
	We call a (randomized) decision rule $\mathcal{R}$ \emph{proportional} if for any $\delta > 0$, there exists $\varepsilon = \varepsilon_{\mathcal{R}}(\delta) > 0$ such that, for any $S \subseteq E$ and $e \in E \setminus S$ with
	% 	monotone if there exists a monotone non-increasing function $\varepsilon\colon (0,1) \to (0,1)$ such that,  and a set $S \subseteq E$,
	\[
	\frac{f_S(e)}{\sum_{e' \in E \setminus S} f_S(e')} > 1 - \varepsilon,
	\]
	the probability that $\mathcal{R}(f,S)$ returns $e$ is at least $1 - \delta$.
	% 	Without loss of generality, we assume that $\varepsilon_p$ is monotone non-increasing in $p$.
\end{definition}
% \ynote{What's between decision rules and monotone decision rules? The following paragraph sounds like every decision rule is a monotone decision rule.}
We note that we can assume $\varepsilon_\mathcal{R}(\delta)$ is non-decreasing  without loss of generality.
To see this, suppose that there exists $\delta < \delta'$ with $\varepsilon_\mathcal{R}(\delta) > \varepsilon_\mathcal{R}(\delta')$.
Then, we can set $\varepsilon_\mathcal{R}(\delta')$ to be $\varepsilon_\mathcal{R}(\delta)$ because if the relative marginal gain of $e$ is more than $1-\varepsilon_\mathcal{R}(\delta) < 1-\varepsilon_\mathcal{R}(\delta')$, then we choose it with probability at least $1-\delta > 1-\delta'$.

We can confirm that \Call{DeterministicGreedy}{} is proportional by setting $\varepsilon = 1/2$ for each $\delta>0$. We also note that \Call{RandomizedGreedy}{} is not proportional, as there is no corresponding $\varepsilon$ for when $1-\delta > \frac{1}{k}$, hence our separate analysis of it.
% However, it can be made proportional by instead of selecting elements from $M$ uniformly, selecting them with probability equal to the proportion of total marginal of elements in $M$.
% Then because $(p_1, \varepsilon_{p_1})$ is a point of this function, we have $(p, \varepsilon_{p_1})$ \ynote{I don't think this is a point in the function. It's just ``covered''} also a point for all $p < p_1$. In particular we have it for $p_2$ also and so in the graph of the function $(p, \varepsilon_p)$, we can replace the point $(p_1, \varepsilon_{p_1})$ with $(p_2, \varepsilon_{p_1})$ and so assume non-increasing.

% Next we give our first lower bound result for general submodular functions, using a monotone decision rule, and we see it is tight, given Lemma~\ref{trivupper}:
\begin{theorem}
	Suppose $k \leq n/2$.
	Then, the sensitivity of any sequential algorithm with a proportional decision rule is $\Omega(k)$.
	% \cnote{this is now just the function from the lower bound}
\end{theorem}
\begin{proof}
	% 	For each natural decision rule $\mathcal{R}$, we provide a monotone submodular function $f\colon 2^E \to \mathbb{R}$ such that the sequential algorithm using $\mathcal{R}$ has a similar behaviour to what is found in Lemma~\ref{trivupper}.
	We assume $n$ is even and fix a proportional decision rule $\mathcal{R}$.
	For $E = \{e_1,\ldots,e_n\}$, define the function $f\colon 2^E \to \mathbb{R}$ as the set function such that
	%\[
	%f(S) = \tilde{f}(1_S),\text{ where }
	%\tilde{f}(x_1,\ldots,x_n) =
	%Ax_1 + B(1-x_1)\sum_{i=n/2+1}^{n}x_i + C\sum_{i = 2}^{n/2}x_i,
	%\]
	\[
	f(S) = Ax_1 + B_i(1-x_1)\sum_{i=n/2+1}^{n}x_i + C\sum_{i = 2}^{n/2}x_i,
	\text{ where } x_i = 1[e_i \in S]\text{ for each }i\in [n].
	\]
	and the constants $A, B_i, C$ are such that $A \gg B_{n/2 +1} \gg \ldots \gg B_n \gg C$.
	In particular, for small enough $\delta = o(1/n)$ and $\varepsilon = \varepsilon_\mathcal{R}(\delta)$, we want
	\begin{align*}
	\frac{A}{A+B(n/2)+C(n/2-1)} > 1-\varepsilon \quad \text{and} \quad \frac{B_i}{\sum_{j=i+1}^n B_j+C(n/2-1)} > 1-\varepsilon.
	\end{align*}
	%\ynote{I think the denominator in the second ratio should be $B(n/2)+C(n/2-1)$. So, there's no single dominating element in $f^{\setminus e_i}$. I guess we need to introduce $B_{n/2} \gg \cdots \gg B_n \gg C$.}
	where the second equation holds for all $i \in \{n/2+1 \ldots n\}$. It is clear that $f$ is monotone submodular. We now analyze the output of the sequential algorithm on $f$ and $f^{\setminus e_i}$ for $i \in [n]$, using the following two claims whose proofs are found in Appendix~\ref{appendix:proofs}:
	
	\begin{claim}\label{claim:PropClaim1}
		We have
		\[
		\Pr\left[|\mathcal{A}(f,k,\mathcal{R}) \cap \{e_1,\ldots,e_{n/2}\}| = k\right] > 1 - \delta.
		\]
	\end{claim}

	%\begin{claim}
	%	For any $j > k$,
	%	\[
	%	\Pr[\mathcal{A}(f^{\setminus e_j}, k,\mathcal{R}) = \{e_1,\ldots,e_k\}] > 0.99.
	%	\]
	%	\ynote{Actually we don't need this claim because we are showing a lower bound.}
	%\end{claim}
	%\begin{proof}
	%	The algorithm will work exactly the same. \ynote{not exactly the same. You need to show that they can be made arbitrarily close. }
	%\end{proof}
	
	%\begin{claim}
	%	For any $1 \leq j \leq k$,
	%	\[
	%	\Pr[ \{e_j,\ldots,e_k\} \cap \mathcal{A}(f^{\setminus e_j}, %k,\mathcal{R}) \neq \emptyset] < 0.01.
	%	\]
	%	\ynote{We only need to consider $e_1$.}
	%\end{claim}
	\begin{claim}\label{claim:PropClaim2}
		We have
		\[
		\Pr\left[|\mathcal{A}(f^{\setminus e_1}, k,\mathcal{R}) \cap \{e_{n/2+1}, \ldots, e_n\}| = k \right] > 1-k\delta .
		\]
	\end{claim}
	We then see that the two claims above immediately show the sensitivity of $\mathcal{A}(\cdot,k,\mathcal{R})$ is $\Omega(k)$.
\end{proof}

\subsection{Randomized Greedy Algorithm}\label{subsec:randomized-greedy}
In this subsection, we analyze the sensitivity of \Call{RandomizedGreedy}{}.
% We recall the definition is given in Algorithm~\ref{alg:randomized}.
\begin{theorem}
	Suppose $k \leq n/4 - 1$. Then, the sensitivity of \Call{RandomizedGreedy}{} is $\Omega(k)$.
\end{theorem}
\begin{proof}
	The function we use is $f: 2^E \rightarrow \mathbb{R}$ with:
	\begin{align*}
	f(S) = Cx_1 + (1-x_i)\sum_{i=2}^{2k+1}x_i + 0.5\sum_{i = 2k+2}^{n}x_i
	\text{ where }x_i = 1[e_i \in S]\text{ for each }i \in [n].
	\end{align*}
	%\ynote{$i$ in $(1-i\varepsilon)(1-x_i)$ is not specified.}
	%\cnote{This was from an earlier definition that was for average sensitivity, it has been deleted}
	where $C$ is large enough so the function remains monotone submodular. We now analyze the output of this algorithm on $f$ and $f^{\setminus e_1}$, with claims whose proofs are found in Appendix~\ref{appendix:proofs}
	\begin{claim}\label{claim:randClaim1}
		We have
		\[
		|\mathcal{A}(f^{\setminus e_1},k,\mathcal{R}) \cap \{e_2,\ldots,e_{2k+1}\}| = k
		\]
	\end{claim}
	
	\begin{claim}\label{claim:randClaim2}
		For $p_i = {\left(\frac{k-1}{k}\right)}^i$, we have
		\begin{equation}
		\Pr\left[|\mathcal{A}(f, k,\mathcal{R}) \cap \{e_{2k+2}, \ldots, e_n\}| = i \right] = p_{k-i-1}/k. \nonumber
		\end{equation}
	\end{claim}
	These two claims enable us to give a simple bound on the sensitivity of this function:
	\begin{align*}
	&\max_{e \in E} d_{\mathrm{EM}}(\mathcal{A}(f), \mathcal{A}(f^{\setminus  e}))
	\geq d_{\mathrm{EM}}(\mathcal{A}(f), \mathcal{A}(f^{\setminus  e_1}))
	\geq \sum_{i=0}^{k-1} p_{k-i-1} \frac{2k-2i}{k} \\
	&=
	2k \left(1 - 2 {\left(\frac{k-1}{k}\right)}^k\right)
	= \Omega(k)
	\end{align*}
	giving us our result.
\end{proof}
\section{Submodular Functions with Bounded Curvature}\label{sec:curvature}
We have now seen that in the general case, a large class of sequential algorithms cannot achieve non-trivial sensitivity in general. With this in mind, we now consider submodular functions with bounded curvature.

In this section, we first begin with another small ``warm-up'' section, where we analyze \Call{DeterministicGreedy}{} and \Call{RandomizedGreedy}{} for small curvature, and show that they also cannot achieve non-trivial sensitivity in general, even when the curvature is arbitrarily close to $0$. The proof techniques we use here will be used in later sections also.
Following on from that, we then move onto general sequential algorithms. We show that under a couple of natural assumptions, no sequential algorithm can achieve non-trivial sensitivity in general.

\subsection{Two Greedy Algorithms}\label{sec:simpleLowCurv}
% We begin by analyzing the sensitivity of \Call{DeterministicGreedy}{} and \Call{RandomizedGreedy}{}.
% \begin{algorithm}[t!]
% 	\caption{Randomized algorithm for monotone submodular function maximization}\label{alg:sequential}
% 	\Procedure{\emph{\Call{Sequential}{$f,k,\mathcal{R}$}}}{
% 		\Input{Monotone submodular function $f\colon 2^E \to \mathbb{R}_+$, an integer $k$.}
% 		$S \gets \emptyset$\;
% 		\While{$|S| \leq k$}{
% 			Let $M$ be a set of size $k$ maximizing $\sum_{e \in M} f(S \cup \{e\})$\;
% 			Select the element $e'$ uniformly from $M$\;
% 			$S \gets S \cup \{e'\}$.
% 		}
% 		\textbf{return} $S$\;
% 	}
% \end{algorithm}

In this section, we prove the following:
\begin{theorem}
	\Call{DeterministicGreedy}{} and \Call{RandomizedGreedy}{} both have sensitivity $\Omega(k)$ sensitivity, even when the input function has curvature arbitrary close to 0.
\end{theorem}
Over the next two subsections, we will prove this statement for each algorithm in turn.
\subsubsection{Deterministic Greedy}
Let $\mathcal{A}$ denote \Call{DeterministicGreedy}{}.
The function $f:2^E \to \mathbb{R}_+$ we claim attains $\Omega(k)$ sensitivity is the following:
\begin{align*}
f(S) = Cx_1 + (1-cx_1)\sum_{i=2}^{k+1} x_i + \left(1-\frac{c}{2}\right)\sum_{i=k+2}^{2k+1} x_i,
\end{align*}
where $x_i = 1[e_i \in S]$, $C$ is a large constant, and $c \in (0,1]$ is the curvature. Note that this function is monotone submodular with curvature $c$. We first analyze the behavior of the algorithm on the full ground set. After selecting the element $e_1$, we see that the algorithm will select the elements $e_{k+1}, \ldots, e_{2k}$ in some order, as the marginals for $e_2, \ldots, e_k$ will be lowered.

Now suppose we delete $e_1$ from the ground set. It is clear the algorithm will select the elements $e_2, \ldots e_{k+1}$ in some order.

Then, we have the sensitivity:
\begin{align*}
\max_{e \in E} |\mathcal{A}(f) \bigtriangleup \mathcal{A}(f^{\setminus  e})|
\geq |\mathcal{A}(f) \bigtriangleup \mathcal{A}(f^{\setminus  e_1})|
= k,
\end{align*}
giving the $\Omega(k)$ lower bound.

\subsubsection{Randomized Greedy}
Let $\mathcal{A}$ denote \Call{RandomizedGreedy}{}.
For this algorithm, we will use a similar function to the previous subsection:
\begin{align*}
f(S) = Cx_1 + (1-cx_1)\sum_{i=2}^{2k+1} x_i + \left(1-\frac{c}{2}\right)\sum_{i=2k+2}^{4k+1} x_i
\end{align*}
where all notation is defined as in the deterministic greedy case. This time, we firstly analyze $\mathcal{A}(f^{\setminus e_1})$. As $e_1$ is deleted, all steps will consist of choosing randomly between elements in the interval ${e_2, \ldots, e_{2k+1}}$.

We then analyze $\mathcal{A}(f)$. In this case we can define a series $p_1, \ldots, p_{k-1}$, where $p_i$ is the probability that $e_1$ will be chosen at the $i$th step. Note that after $e_1$ is chosen, all further elements will be chosen from ${e_{2k+2}, \ldots, e_{4k+1}}$.
% In particular, note that the elements $\{e_2,\ldots,e_{2k+1}\}$ will never be in the output set of $\mathcal{A}(f^{\setminus e_1})$.

We therefore regard that if $e_1$ is chosen at the $i$th step, then the remaining $k-i$ steps will never happen in $\mathcal{A}(f^{\setminus e_1})$. We also note that $p_i = {\left(\frac{k-1}{k}\right)}^{i-1}\left(\frac{1}{k}\right)$. We can now derive a lower bound for the sensitivity as follows:
\begin{align*}
&\max_{e \in E} d_{\mathrm{EM}}(\mathcal{A}(f), \mathcal{A}(f^{\setminus  e}))
\geq d_{\mathrm{EM}}(\mathcal{A}(f), \mathcal{A}(f^{\setminus  e_1}))
\geq \sum_{i=1}^{k-1}p_i \left(k-i\right) \\
&= \sum_{i=1}^{k-1}{\left(\frac{k-1}{k}\right)}^{i-1}\left(\frac{1}{k}\right) \left(k-i\right)
= \Omega(k)
\end{align*}
where the second inequality comes from the fact that a lower bound comes from only considering the difference of probability in choosing elements from $e_{2k+2}, \ldots, e_{4k+1}$. This gives the required lower bound.
\subsection{General Hardness Results}\label{subsec:GenHard}
In this section, we detail some hardness results for general $k$ and $c$, and show that under a couple of small assumptions on our algorithm that are quite reasonable, for any algorithm $\mathcal{A}$ there is a function $f$ such that the sensitivity of $f$ is $\Omega(k)$, even in the case of bounded curvature.

%\ynote{Specify this is about the input function or the algorithm.}

\subsubsection{Algorithm Formalization}
In Algorithm~\ref{alg:genseq}, we give a framework for a sequential algorithm where elements are selected one at a time, that is a special case of Algorithm~\ref{alg:sequential}, but includes many common algorithms like \Call{DeterministicGreedy}{} and \Call{RandomizedGreedy}{}. In this algorithm, the sets $S_{k,i} \subseteq \mathbb{N}$ give the indices of a number of elements at step $i \in [k]$ that may have non-zero probability assigned to them, where the indices refer to the ordinal position of the elements marginal.

We also make the assumption that the probabilities assigned to elements are independent of marginals, and only depend on the ordinal position of the marginals. This is what motivates our choice of name, \Call{IndependentSequential}{}. Therefore, our algorithm can be specified by sets $\{S_{k,i}\}_{k\in \mathbb{N},i \in [k]}$, and distributions over the positive integers $\{D_{k,i}\}_{k\in \mathbb{N},i \in [k]}$, where $D_{k,j}$ is the distribution on the positive integers at step $j$, when the cardinality constraint is $k$ steps. At each point, the algorithm will choose an element from $S_{k,i}$ with respect to the distribution $D_{k,i}$, as illustrated in Algorithm~\ref{alg:genseq}. We also assume that the sets $S_{k,i}$, are defined solely via $k$ for each $i$, and do not vary as $n$ varies. As an example, this means we don't permit sets like $S_{k,i} = \{1, \ldots, n/3\}$.
%\ynote{Is this $i$ the $k$ specified by the cardinality constraint?}.

% \ynote{So $S_{k,i}$ itself is not randomized though the process of selecting an element from $S_{k,i}$ could be randomized?} \cnote{Yes, even if $S_{k,i}$ is randomized, you can still just list everything in the domain as a deterministic $S_{k,i}$ and proceed as normal}.
The randomized decision rule $\mathcal{R}$ assigns probability of selecting the elements in $S_{k,i}$ according to the distribution $D_{k,i}$, and then selects an element according to those probabilities. %\ynote{$h$ in the algorithm description should be $\mathcal{R}$? Also, no explanation about $\mathcal{R}$ (or $h$).}
% \ynote{How is this different from Algorithm~\ref{alg:sequential}? Special case?} \ynote{Also the name sounds inconsistent: \textsc{Sequential} vs \textsc{SequentialAlgorithm}}
% \cnote{Yes this is a subcase of Algorithm~\ref{alg:sequential}, as it requires the independence of actual marginal values. Not sure if IndependentSequential works here - wanted to name the algorithm in such a way that it implies that property we've assumed in the name. }
\begin{example}
	For \Call{DeterministicGreedy}{}, we have $S_{k,i} = \{1\}$ for any $k \in \mathbb{N}$ and $i\in [k]$. We also have the distributions $D_{k,i}$ assigning probability $1$ to the integer $1$, for all $k$ and $i$.
	We then have $\mathcal{R}(S_{k,i}, D_{k,i})$ returns $e_1$ with probability $1$, where the ordering of the elements is the same as the ordering of the marginal.
	
	For \Call{RandomizedGreedy}{}, we have $S_{k,i} = \{1,\ldots,k\}$ for any $k \in \mathbb{N}$. The distribution $D_{k,i}$ assigns equal probability to all integers from $1$ to $k$. We then have $\mathcal{R}(S_{k,i}, D_{k,i})$ assigns the probability of selecting each of $e_1, \ldots, e_k$ as $1/k$, then chooses one and returns it.
\end{example}
% \cnote{It seems like for the simpler proof, the discussion about the tiebreak isn't needed at all}.
%\begin{enumerate}
%\item
%\item Additionally $S_{k,i}$ is constant for each $i$, so while the indices might be in terms of $k$ and so the number of elements may vary with $k$, the definition will not. \cnote{you need to phrase this a better way}
%\item There is a consistent tie-break scheme for deciding what elements feature in $S_{k,i}$, which we assume is a total order on $\{1, \ldots, n\}$. Without loss of generality \ynote{Is this assumption or without loss of generality?} \cnote{I guess it is an assumption? I was supposing that the tiebreak scheme consists of a total order on $\{1,\ldots,n\}$, and then WLOG supposing that the total order was $1 \geq 2 \geq \ldots \geq n$}, we can suppose that the tiebreak selects elements with lower index in the function definition, that is, the total order is $1 \geq 2 \geq \ldots \geq n$. \cnote{I'll check again but not sure I need this at all now that we're using the simpler proof}
%\end{enumerate}
\begin{algorithm}[t!]
	\caption{Independent sequential algorithm}\label{alg:genseq}
	\Procedure{\emph{\Call{IndependentSequential$_{\mathcal{S},\mathcal{D}}$}{$f,k$}}}{
		\Comment{This algorithm is parameterized by a family of sets $\mathcal{S} = \{S_{k,i} \subseteq \mathbb{N} \}_{k \in \mathbb{N}, i \in [k]}$ and a collection of distributions $\mathcal{D} =  \{D_{k,i} \subseteq \mathbb{N} \}_{k \in \mathbb{N}, i \in [k]}$}\;
		\Input{Monotone submodular function $f\colon 2^E \to \mathbb{R}_+$ and an integer $k$.}
		$S \gets \emptyset$\;
		\While{$|S| \leq k$}{
			Sort elements according to their marginal, with $e_1$ the largest, and select the elements corresponding to $S_{k,i}$\;
			$e \leftarrow \mathcal{R}(S_{k,i}, D_{k,i})$\;
			$S \gets S \cup \{e\}$.
		}
		\textbf{return} $S$\;
	}
\end{algorithm}

The goal of this section is to prove the following:
\begin{theorem}\label{mainhard}
	For every algorithm satisfying the assumptions above, there exists a function $f$ that attains sensitivity $\Omega(k)$, or it does not attain a constant-factor approximation.
\end{theorem}
\subsubsection{Hardness Preliminaries}
Our proof will rely on the fact that we can design a function where one particular element must be chosen with probability $\Omega(1)$ relatively early on in the process of the algorithm. For an element $e \in E$, step $i \in [k]$, and current solution set $S \subseteq E$, let $p_{i,S}(e)$ denote the probability of selecting the element $e$ at step $i$ when the current solution set is $S$. Then if we let $q(S)$ be the probability that $S$ is the solution set after $|S|-1$ steps, we can define the following two quantities:
\begin{align*}
p_i(e) = \sum_{S \subseteq E:|S|=i-1} p_{i,S}(e)q(S), \quad P_i(e) = \sum_{j = 1}^i \prod_{k=1}^{j-1}(1-p_j(e))p_i(e),
\end{align*}
%and from there define:
%\begin{align*}
%P_i(e) = \sum_{j = 1}^i \prod_{k=1}^{j-1}(1-p_j(e))p_i(e)
%\end{align*}
so we have that $p_i(e)$ is the probability of selecting $e$ at step $i$, and $P_i(e)$ is the probability of having selected $e$ by step $i$.
For an element $e^* \in E$, we similarly define $p^{\setminus e^*}_{i,S}(e), p^{\setminus e^*}_i(e), P^{\setminus e^*}_i(e)$ as the same quantities for when the algorithm is run on the ground set $E \setminus \{e^*\}$.

The result we will use for proving Theorem~\ref{mainhard} is the following, the proof of which is found in Appendix~\ref{appendix:proofs}:
\begin{lemma}\label{lem:hardmain}
	Let $\mathcal{A}$ be an algorithm, $f:2^E \to \mathbb{R}_+$ be a monotone submodular function, $e^* \in E$ be an element.
	Suppose we have $P_{i^*}(e^*) = \Omega(1)$ for some $i^* \in [k]$ with $k-i^* = \Omega(k)$.
	Additionally, suppose there exist subsets of elements $E_1 = \{e_1, \ldots, e_{a-1}\}$ and $E_0 = \{e_a,e_{a+1},\ldots,e_b\}$ with $e^* \in E_1$ such that $\mathcal{A}(f^{\setminus e^*})$ only selects elements from $E_1$, and $\mathcal{A}(f)$ only selects from $E_1$ before $e^*$ is chosen, and only selects elements from $E_0$ after $e^*$ is chosen.
	Then the function $f$ attains $\Omega(k)$ sensitivity for the algorithm $\mathcal{A}$.
	%Additionally, suppose we have the following for all $i \leq i^*$, for two sets of elements, with $a_2 > b_1$, and $a_1 > j$:
	%\begin{align}
	%\sum_{j=a_1}^{b_1}\left(P_k(e_j) - P_{i+1}(e_j)\right) = \Omega(k) \nonumber \\
	%\sum_{j=a_2}^{b_2}\left(P_k(e_j)\right) = 0 \nonumber \\
	%\sum_{j=a_1}^{b_1}\left(P^{\setminus e^*}_k(e_j) - P^{\setminus e^*}_{i+1}(e_j)\right) = 0 \nonumber \\
	%\sum_{j=a_2}^{b_2}\left(P^{\setminus e^*}_k(e_j)\right) = \Omega(k) \nonumber
	%\end{align}
	%where we have $m_{b_2}/m_{a_1} \geq (1-c)$, where $c$ is the curvature. Then the function $f$ attains $\Omega(k)$ sensitivity.
\end{lemma}

There are two main stages to our proof of Theorem~\ref{mainhard}. In the first, we consider the performance of our algorithm on a modular function with a large-weight element, and derive conditions on the distributions $D_{k,1}, \ldots, D_{k,k}$ to ensure we get a constant-factor approximation.
We will call this the \textit{large-element scheme}. After that, we look at the function whose exact form will be motivated and described later, but intuitively will have all elements very close together. We call this the \textit{near-equality scheme}. We will show that the conditions derived in the large-element scheme can be used to show that the function in the near-equality scheme attains high sensitivity.

\subsubsection{Large Element Scheme}
We consider the following function:
\begin{align*}
f(S) = \sum_{i=1}^k x_i + \varepsilon\sum_{i=k+1}^n x_i, \text{ where } x_i = 1[e_i \in S]  % \label{eqn:largeF}
\end{align*}
and suppose that our algorithm $\mathcal{A}$ has a constant factor approximation $\alpha > 0$.
This means that at step $(1 - \alpha/2)k$, we know that at least a $(\alpha/2)$-fraction of the elements $e_1, \ldots, e_k$ must be chosen.

Now, as in Lemma~\ref{lem:hardmain}, let $P_i(e)$ be the probability of selecting element $e$ before step $i$ (inclusively).
% Knowing that at least a $(c/2)$-fraction elements from $e_1, \ldots, e_k$ must be chosen by step $(1 - \alpha/2)k$, we see:
Then, we have
\begin{align}
\sum_{i=1}^k P_{(1-\alpha/2)k}(e_i) \geq \frac{\alpha k}{2} \label{equation:largescheme}
\end{align}
and therefore $\Omega(k)$ of these elements must be chosen by this step with probability $\Omega(1)$ to ensure a constant factor approximation. Knowing this, we can now move on to the near-equality scheme.

\subsubsection{Near-Equality Scheme}
Recall that we have just seen that $\Omega(k)$ elements between $e_1, \ldots, e_k$ must be chosen with probability $\Omega(1)$ to ensure a constant-factor approximation in the large-element scheme. Recall that our choices of elements is determined solely by the ordering of the marginals, but not the actual marginal values themselves.

Therefore, as the large-element function was modular, we can say that for each of those $\Omega(k)$ elements $e_j$, we can give a range of other functions $f$ that will also have $e_j$ chosen with probability $\Omega(1)$ with $\Omega(k)$ steps remaining, as long as that function $f$ can preserve the same ordering of elements before $e_j$ is chosen.

%\ynote{This is a sloppy argument. We do not a priori don't know how the algorithm works for different functions.} \cnote{That is true, but as the distributions $D$ only depend on the ordering of the marginals, so long as we keep the ordering the same (before $e_j$ is chosen) we can freely change the marginal values and keep the same behaviour on $e_j$, namely that it will be chosen with constant probability while $\Omega(k)$ steps are remaining, which I've tried to explain.}
Now define $f:2^E \to \mathbb{R}_+$ as follows:
% \ynote{Can't we directly show what we showed in the ``Large Element Scheme'' section for the function $f$ below? The current argument is fine, but unnecessarily indirect...}
% \cnote{I don't think so - the purpose of above was to show an element was chosen with $\Omega(1)$ probability while there are $\Omega(k)$ steps remaining - which worked for that function as there were only $\mathcal{O}(k)$ large elements, and so $\mathcal{O}(k)$ entries in the summation above. But here we could have $I_{max} = \omega(k)$ in general. Even though we have a total of $\Omega(k)$ probability mass to assign to large elements, because there could be so many large elements we cannot in general say that there is any one element that is chosen with $\Omega(1)$ probability. Following the previous subsection we would have the same equation, but instead summing over $\omega(k)$ elements and so we couldn't draw our conclusion. This was my reasoning for coming up with the two separate functions at all. I've left a brief note in the text below.}.
\begin{align}
f(S) &= \left((1-ce_j)\sum_{i=1}^{j-1} e_i\right) + e_j + \left((1-c + \varepsilon(1-e_j))\sum_{i=j+1}^{I_{\max}+1}e_i\right) \nonumber \\
&+ (1-c + \varepsilon/2)\sum_{i=I_{\max}+2}^{2I_{\max}+1} e_i + \varepsilon\sum_{2I_{\max}+2}^{n} e_i \nonumber
\end{align}
where $x_i = 1[e_i \in S]$ for each $i \in [n]$ and $I_{\max}$ is the element with largest index across all of the sets $\{S_{k,i}\}$. Note that for all $S \in V$ such that $e_j \not\in S$, that the ordering of all elements in $V \setminus S$ by marginal value at the point $f(S)$ is the same for the functions in the near-equality and large-element scheme. Therefore, we can say that for this function also, $e_j$ will be chosen with probability $\Omega(1)$ with $\Omega(k)$ steps remaining. We briefly note that we couldn't simply use this function to reach the conclusion in the previous subsection directly, as we may have $I_{\max}$ is $\omega(k)$, and so in Equation~\eqref{equation:largescheme} we would have $\omega(k)$ elements in the summation, meaning none would necessarily be selected with $\Omega(1)$ probability.

Note that before $e_j$ is selected, all elements that are selected come from the set $\{e_1, \ldots, e_{I_{\max}}\}$. After selecting $e_j$, the top $I_{\max}$ elements ordered by marginal will be $\{e_{I_{\max}+1}, \ldots, e_{2I_{\max}}\}$, and so all further elements that must be chosen (of which there are $\Omega(k)$), must be chosen from this set. Then if we consider the function ran on the reduced ground set $E\setminus\{e_j\}$, then we see that all large elements chosen will come from the set $\{e_1,\ldots,e_{I_{\max}+1}\}\setminus \{e_j\}$, by the fact that the resulting function is modular, and the definition of $I_{\max}$.

This gives us the statement for Theorem~\ref{mainhard} using Lemma~\ref{lem:hardmain}, noting we can set $e^* = e_j$, $E_1 = \{e_1, \ldots, e_{I_{\max}+1}\}$, $E_0 = \{e_{I_{\max}+2}, \ldots, e_{2I_{\max}+1}\}$.

\begin{algorithm}[t!]
	\caption{Proportional greedy algorithms}\label{alg:propgreedy}
	\Procedure{\emph{\Call{ProportionalGreedy}{$f,k$}}}{
		\Input{Monotone submodular function $f\colon 2^E \to \mathbb{R}_+$ and an integer $k$.}
		$S \gets \emptyset$\;
		\While{$|S| \leq k$}{
			Select $e \in E \setminus S$ with probability proportional to the marginal gain $f_S(e)$\;
			$S \gets S \cup \{e\}$.
		}
		\textbf{return} $S$\;
	}
\end{algorithm}
%INFORMAL IDEA:
%Note there is at least $(1-d)k$ probability mass assigned to small elements between $dk+1, B$. We claim however, that before $\Omega(k)$ of this probability mass is assigned, that we must assign $\Omega(k)$ of the probability to large elements first, which we can see by changing $k$ to $k_0$ and driving down $k$, at least so its smaller than $(1-d)k$. Then we can use the fact that the functions $S_{k,i}$ are constant to show this.
%The average sensitivity counterpart result can be found in Appendix \ref{appendix:AveGeneral}. The statement is the same, but the new trivial upper bound that is matched is $\mathcal{O}(k^2/n)$.

\subsubsection{Obtaining bounded sensitivity}
Theorem~\ref{mainhard} makes it difficult to identify a situation where we can guarantee non-trivial sensitivity upper bounds.
As one example however, we can show that in the regime $k = \Omega(n)$, we can find this for an algorithm we call \Call{ProportionalGreedy}{} (see Algorithm~\ref{alg:propgreedy}). However, while this algorithm does not have a constant-factor approximation in general, though in the regime $k = \Omega(n)$ we do get the following, with the proof in Appendix~\ref{appendix:PropGredFactor}:
\begin{theorem}
	If we have $k \geq cn$ for some constant $c$, then the Algorithm \Call{ProportionalGreedy}{} has a constant factor approximation of $\left(1 - \frac{1}{e^{c/(1-c)}}\right)$
\end{theorem}
Our result on the sensitivity, proved in Appendix~\ref{appendix:PropGreedyproof} is as follows:
\begin{theorem}
	When we have $k = \Omega(n)$, then \Call{ProportionalGreedy}{} has a constant factor approximation and achieves the following sensitivity bound when the curvature is $c$:
	\begin{align*}
	O\left(\frac{{\left(1 - \sqrt{1-c}\right)}^2}{c} \cdot k\right).
	\end{align*}
	Further, this bound is tight.
\end{theorem}

\section{Distributed Algorithms}\label{sec:distributed}
We have seen that many of our hardness results rely on one element changing the ordering of marginals greatly. Therefore one could hope that splitting the ground set among many machines would enable us to obtain a non-trivial upper bound for all $k, n$, as that one element can't affect too much of the ground set overall. However, this is not necessarily the case.

For both the algorithm \Call{GreeDi}{} found in~\cite{barbosa2015power}, and for a general framework found in~\cite{barbosa2016new} for constructing distributed algorithms for centralized ones, we see that we still obtain high sensitivity. Details of \Call{GreeDi}{} and the general framework, along with the proofs of the following results, are found in Appendix~\ref{appendix:distributed}:
\begin{theorem}\label
	The algorithm \Call{GreeDi}{} attains $\Omega(k)$ sensitivity for some function for sufficiently large $n$, for all curvature $0 < c \leq 1$.
\end{theorem}
\begin{theorem}
	Let $\mathcal{A}$ be some algorithm with $\Omega(k)$ sensitivity, as shown by Theorem~\ref{mainhard}. Then $\mathcal{A}_d$, a distributed algorithm constructed via the framework in~\cite{barbosa2016new} from $\mathcal{A}$ also has $\Omega(k)$ sensitivity for some function.
\end{theorem}

\section{Conclusion}\label{sec:conclusion}

While we have given some strong negative results about sensitivity for monotone submodular maximization with a cardinality constraint, some obvious questions remain---namely does there exist any constant-factor algorithm for monotone submodular maximization that has a non-trivial sensitivity bound? As we have seen, it looks unlikely when using discrete algorithms, and so continuous algorithms would be the most promising.
%However, we also saw that the continuous greedy algorithm behaves almost like the discrete greedy when one marginal is very large, and so it appears we can achieve $\Omega(k)$ sensitivity in this case.

We can also ask if there other settings in which our results imply other hardness bounds? We have seen that our centralized results apply to distributed algorithms, but as example it would be interesting to know what, if anything, our results imply in the streaming setting?

A final, and important question is, does there exist a relevant subclass of submodular functions such that we can obtain a non-trivial sensitivity bound? From what we have seen, it doesn't seem like that subclass is bounded curvature. An alternative could be $M^{\natural}$-concave functions~\cite{murota2010submodular}, which are submodular functions equipped with an exchange property. A sufficient question to answer would be: is the solution to the weighted assignment problem stable when deleting an element? We note that the argument for the minimum spanning forest problem in~\cite{varma2019average} can be used to show that the matroid rank function has $O(1)$ sensitivity. As the unweighted assignment problem can be described by a matroid rank function, we see that $M^{\natural}$-concave functions form an interesting ``midpoint'' class that could warrant further investigation.

\newpage
\bibliographystyle{plain}
\bibliography{AverageSensitivity}
\newpage
\appendix

\section{Proofs of missing results}\label{appendix:proofs}
\subsection{Claim~\ref{claim:PropClaim1}}
\begin{proof}
	From the choice of $A$, we choose $e_1$ with probability at least $1-\delta$ in the first step.
	When this happens, the marginal gains of $e_i$ for $i > n/2$ are $0$ in subsequent steps, and hence we choose only elements $e_i$ for $i \leq n/2$.
	% 			Then the returns for all elements multiplied by $B$ will be $0$ and so we choose only elements $e_i$ with $i > n/2$.
\end{proof}
\subsection{Claim~\ref{claim:PropClaim2}}
\begin{proof}
	% 			We don't need to consider $A$ in this case.
	% We do not need to consider the first term in $\tilde{f}$.
	By the choice of $B_i$ and $C$, we see that we successively choose the elements $B_{n/2+1}, \ldots, B_n$ with probability at least $1-\delta$. We then see that the probability of selecting each of these elements in turn is at least $(1-\delta)^k$, and using Bernoulli's inequality gives our result.
	%\ynote{need to apply a union bound over $k$ steps.} %\cnote{todo}
\end{proof}
\subsection{Claim~\ref{claim:randClaim1}}
\begin{proof}
	We see that $p_i$ is clearly the probability that $e_1$ is not selected after $i$ iterations. Additionally, we see after selecting $e_1$, that all further elements under consideration will be from the set $e_{2k+2}, \ldots, e_n$, giving our claim.
\end{proof}
\subsection{Claim~\ref{claim:randClaim2}}
\begin{proof}
	As we select from the top $k$ marginals, at the first step we will select one of $e_2, \ldots, e_{k+1}$, and then add $e_{k+2}$ to the set of elements that we could possibly choose.
	
	After $k$ steps, the element with maximum index we could include is $e_{2k+1}$, giving our claim.
\end{proof}

\subsection{Lemma~\ref{lem:hardmain}}
\begin{proof}
	Note that the sensitivity can be bounded from below by the following total variation distance between the probabilities of selecting elements $e_i$ in $\mathcal{A}(f)$ and $\mathcal{A}(f^{\setminus e^*})$:
	\begin{align*}
	d_{\mathrm{EM}}(\mathcal{A}(f), \mathcal{A}(f^{\setminus e^*})) \geq \sum_{i=1}^n |P_k(e_i) - P_k^{\setminus e^*}(e_i)|
	\end{align*}
	we can again bound this from below by restricting the summation to elements in the set $E_0$. Then noting that $\Omega(k)$ steps are still taken when we start selecting from $E_0$, and that $P_{i^*}(e^*)$ is $\Omega(1)$, we have:
	%and letting $\mathcal{I}^*$ be the event of $e^*$ being selected at step $i^*$:
	\begin{align*}
	\sum_{i=1}^n |P_k(e_i) - P_k^{\setminus e^*}(e_i)| \geq \sum_{e \in E_0} |P_k(e) - P_k^{\setminus e^*}(e)| \geq P_{i^*}(e^*) \cdot (k-i^*) = \Omega(k)
	\end{align*}
	as required.
	%\cnote{Not sure if this entirely makes sense. As with $\Omega(1)$ probability we have to choose $\Omega(k)$ more elements, the above sum is $\Omega(k)$ - Initially I had tried to condition on the event of $e^*$ being selected by step $i^*$ in the second term (which was $|P_k(e)| | I^*$) but not sure how this is best phrased.}
	%\ynote{What's $|P_k(e)| | I^*$?}
	%Define distributions on the ground set elements $Q(e) = P_k(e) - P_{i+1}(e)$ (and similarly $Q^{\setminus e^*}(e)$, which gives the probabilities that each element will be selected somewhere between the $i+1$th and $k$th step. Note that for $a_2 \leq j \leq b_2$, we have $Q(e_j) = 0$.
	%A lower bound for the sensitivity can then be calculated as:
	%\begin{equation}
	%P_{i^*}(e^*) \sum_{e} |Q(e) - Q^{\setminus e^*}(e)| \geq P_{i^*}(e^*) \sum_{j=a_1}^{b_2} |Q(e_j) - Q^{\setminus e^*}(e_j)| = \Omega(k)
	%\end{equation}
	%where the equality comes from the Equations in the hypotheses of the lemma, and the definition of $Q, Q^{\setminus e^*}$.
\end{proof}

\section{Proof of Randomized Greedy Algorithm}\label{appendix:PropGredFactor}
In this section, we prove the following:
\begin{theorem}
	\Call{ProportionalGreedy}{} in Algorithm~\ref{alg:propgreedy} is a $(1-e^{-c/(1-c)})$-approximation algorithm when $k \geq cn$ for $c \in [0,1]$.
	%     That is, for the output $S \subseteq E$ and the optimal set $\mathrm{OPT} \subseteq E$, we have
	% 	\begin{equation}
	% 		\mathbb{E}(f(S)) \geq (1-e^{-1})f(\mathrm{OPT})
	% 	\end{equation}
\end{theorem}
\begin{proof}
	The proof will be very similar to the one for the  deterministic greedy algorithm.
	Let the optimal set $\mathrm{OPT} = \{e_1, \ldots, e_k\}$, where the elements are in arbitrary order (and we can assume $|\mathrm{OPT}|=k$ by monotonicity).
	We denote the partial solution after $i$ steps as $S_i$.
	Note that $S_i$ is a random variable.
	By the random selection rule, we have
	\begin{align*}
	\mathbb{E}[f(S_{i+1}) - f(S_i)] = \frac{1}{M}\sum_{e \in E \setminus S_i} {f_{S_i}(e)}^2
	\end{align*}
	where $M = \sum_{e \in E \setminus S_i} f_{S_i}(e)$.
	Then, We have
	\begin{align*}
	f(\mathrm{OPT}) & \leq
	f(\mathrm{OPT} \cup S_i) \tag{by monotonicity} \\
	&= f(S_i) + \sum_{j=1}^k f_{S_i \cup \{e_1, \ldots, e_{j-1} \}}(e_j) \\
	& \leq f(S_i) + \sum_{j=1}^k f_{S_i}(e_j). \tag{by submodularity} \\
	& \leq f(S_i) + M \tag{by $\sum_{j=1}^k f_{S_i}(e_j) \leq M$} \\
	& \leq f(S_i) + (n-k)\mathbb{E}[f(S_{i+1}) - f(S_i)]. \tag{by $M^2 \leq (n-k)\sum_{e \in E \setminus S_i} f_{S_i}(e)^2$.}
	\end{align*}
	however, in our case we have $k = cn$, so we can rearrange this to:
	\begin{equation}
	f(\mathrm{OPT}) \leq f(S_i) + \left(\frac{1}{c}-1\right)k\mathbb{E}[f(S_{i+1}) - f(S_i)]
	\end{equation}
	Rearranging the inequality, and letting $k' = \left(\frac{1}{c}-1\right)k$ we obtain, as in the deterministic greedy:
	\begin{align*}
	f(\mathrm{OPT}) - f(S_i) \leq k'\mathbb{E}[f(S_{i+1}) - f(S_i)]
	\end{align*}
	% 	Now we can unfix the event $A_i$ and take an expectation over the set of events, which gives us:
	Following the proof of the deterministic greedy algorithm, we then get:
	\begin{align*}
	\mathbb{E}[f(S_{k})] \geq \left(1 - {\left(1- \frac{1}{k'}\right)}^k\right)f(\mathrm{OPT}).
	\end{align*}
	From here, we can use the fact that $k' = \left(\frac{1}{c}-1\right)k$ and deterministic limits to show:
	\begin{align*}
	\mathbb{E}[f(S_{k})] \geq \left(1 - \frac{1}{e^{c/(1-c)}}\right) f(\mathrm{OPT})
	\end{align*}
	as required.
	%It follows that
	%\begin{align*}
	%	f(\mathrm{OPT}) - \mathbb{E}[f(S_i)] \leq k'\mathbb{E}[f(S_{i+1}) - f(S_i)].
	%\end{align*}
	%Upon rearranging and applying recursion, we obtain
	%\begin{align*}
	%		f(\mathrm{OPT}) - \mathbb{E}[f(S_{i+1})] \leq {\left(1 - \frac{1}{k'}\right)}^i\left(f(\mathrm{OPT}) - \mathbb{E}[f(S_0)]\right)
	%	= {\left(1 - \frac{1}{k'}\right)}^i f(\mathrm{OPT}).
	%\end{align*}
	% 	Using $S_0 = 0$ and the approximation for the exponential we get:
	%Hence, we have
	%\begin{align*}
	%		\mathbb{E}[f(S_{k})] \geq \left(1 - \left(1- \frac{1}{k'}\right)^k\right)f(\mathrm{OPT}) \geq \left(1 - \frac{1}{e}\right)f(\mathrm{OPT}),
	%\end{align*}
	%as desired.
\end{proof}
Note that as $c \rightarrow 0$ (which brings us to the general case) that the approximation factor tends to $0$.

\section{The Proportional Greedy Algorithm}\label{appendix:PropGreedyproof}
%\cnote{Change intro}
In Section~\ref{sec:sequential}, we showed that the trivial upper bound on sensitivity is tight for proportional sequential algorithms.
We also showed that for a large class of constant factor approximation algorithms, where the probability of selecting an element depends only on its ordinal position among marginals, that the trivial upper bound on sensitivity is also tight.

Hence in this section, we introduce an algorithm that only has a constant-factor approximation when $k = \Omega(n)$, but attains non-trivial sensitivity for all functions.

\subsection{Upper Bound for the Randomized Greedy Algorithm}

In this section, we prove the following:
\begin{theorem}\label{thm:ub-randomized-greedy}
	The sensitivity of the randomized greedy algorithm is
	\[
	% \frac{c(k-1)}{{(1+\sqrt{1-c})}^2} + 1
	O\left(\frac{{(1-\sqrt{1-c})}^2}{c}\cdot k \right).
	\]
	In particular, the sensitivity is $O(1)$ when $c = O(1/k)$.
\end{theorem}
Let $\mathcal{A}$ denote the randomized greedy algorithm.
Let $\mathcal{I}$ denote one iteration of $\mathcal{A}$, that is, given a function $f\colon 2^E \to \mathbb{R}$, it outputs $e \in E$ with probability $f(e)/ \sum_{e' \in E}f(e')$.
% Note that $\mathcal{A}(f) = S_k$, where $S_k = S_{k-1} + \mathcal{I}(f_{S_{k-1}}), S_{k-1} = S_{k-2} + \mathcal{I}(f_{S_{k-2}}),\ldots,S_1 = \{\mathcal{I}(f)\}$.

In what follows, we fix a monotone submodular function $f\colon 2^E \to \mathbb{R}$ and let $c \in [0,1]$ be the curvature of $f$.
Let $e^* \in E$ be the element that defines the sensitivity of $\mathcal{A}$, that is, $d_{\mathrm{EM}}(\mathcal{A}(f), \mathcal{A}(f^{\setminus e}))$ is maximized when $e = e^*$.
Letting $f' = f^{\setminus e^*}$, our goal is to bound $d_{\mathrm{EM}}(\mathcal{A}(f), \mathcal{A}(f'))$ from above.

The following lemma shows that, given that $e^*$ is already selected in the process of $\mathcal{A}$ on $f$, the earth mover's distance between $\mathcal{A}$ on $f$ and $f'$ does not increase by much in one iteration.
\begin{lemma}\label{lem:em-one-iteration}
	For any $S,T \subseteq E - e^*$, we have
	\[
	d_{\mathrm{EM}}\bigl(S + \mathcal{I}(f_{S + e^*}),T + \mathcal{I}(f'_T)\bigr) \leq |S\bigtriangleup T| + \frac{{\left(1-\sqrt{1-c}\right)}^2}{c}.
	\]
\end{lemma}
\begin{proof}
	Let $e_f = \mathcal{I}(f_{S+e^*})$ and $e_{f'} = \mathcal{I}(f'_{T})$.
	Note that $e_f$ and $e_{f'}$ are random variables and we assume they are independent.
	Conditioned on the event $e_f \in T \setminus S$ or the event $e_{f'} \in S \setminus T$, we have
	\[
	\bigl|(S + e_f) \bigtriangleup (T + e_{f'})\bigr|\leq |S\bigtriangleup T|.
	\]
	Note that the condition can be rephrased as $e_f \not \in R$ or $e_{f'} \not \in R$, where $R = E \setminus ((S \cup T) +e^*)$.
	Let $e_f^R$ (resp., $e_{f'}^R$) be $e_f$ (resp., $e_{f'}$) conditioned on that $e_f \in R$ (resp., $e_{f'} \in R$).
	Then, we have
	\begin{align*}
	& d_{\mathrm{EM}}(S + e_f,T + e_{f'}) \\
	& \leq
	\Pr[e_f \not \in R \vee e_{f'} \not \in R] \cdot |S \bigtriangleup T|
	+ \Pr[e_f \in R \wedge e_{f'} \in R] \cdot        d_{\mathrm{EM}}\bigl(S + e_f^R,T + e_{f'}^R \bigr) \\
	& \leq
	\Pr[e_f \not \in R \vee e_{f'} \not \in R] \cdot |S \bigtriangleup T|
	+
	\Pr[e_f \in R \wedge e_{f'} \in R] \cdot |S \bigtriangleup T| \\
	& \quad + \Pr[e_f \in R \wedge e_{f'} \in R] \cdot  d_{\mathrm{TV}}\bigl(e_f^R , e_{f'}^R \bigr) \\
	& \leq |S \bigtriangleup T| + d_{\mathrm{TV}}\bigl(e_f^R, e_{f'}^R\bigr).
	\end{align*}
	
	To prove the lemma, it suffices to show that $d_{\mathrm{TV}}\bigl(e_f^R, e_{f'}^R\bigr) \leq {(1-\sqrt{1-c})}^2/c$.
	To this end, define vectors $\bm{m},\bm{m}' \in \mathbb{R}^R$ as $\bm{m}(e) = f_{S+e^*}(e)$ and $\bm{m}'(e) = f'_T(e)$ for each $e \in R$.
	Then, we have $(1-c)\bm{m}'(e) \leq \bm{m}(e) \leq \bm{m}'(e)$ for every $e \in R$, where the first inequality is due to $f$ having curvature $c$ and the second inequality is due to the submodularity of $f$.
	Then, we have
	\begin{align*}
	& d_{\mathrm{TV}}\bigl(e_f^R, e_{f'}^R\bigr)
	\leq \max_{\bm{\alpha} \in {[1-c,1]}^R}h(\bm{\alpha}),
	\text{ where } h(\bm{\alpha}) :=
	\frac12 \sum_{e \in R} \left|\frac{\bm{\alpha}(e) \bm{m}'(e)}{\bm{\alpha}^\top \bm{m}'} - \frac{\bm{m}'(e)}{\bm{1}_R^\top \bm{m}'}  \right|.
	\end{align*}
	In what follows, we compute the maximum value of $h$.
	\begin{claim}
		The function $h$ is maximized by some $\bm{\alpha} \in {\{1-c,1\}}^R$.
	\end{claim}
	\begin{proof}
		For $e \in R$, let $\sigma_e = 1$ if $\frac{\bm{\alpha}(e) \bm{m}'(e)}{\bm{\alpha}^\top \bm{m}'} \geq \frac{\bm{m}'(e)}{\bm{1}_R^\top \bm{m}'}$ and $\sigma_e = -1$ otherwise.
		Then, the subdifferential of $h$ with respect to $\bm{\alpha}(e)$ includes
		\begin{align*}
		\frac{\partial{h}}{\partial \bm{\alpha}(e)}
		& \ni
		\frac{1}{2} \sigma_e \frac{\bm{m}'(e)(\bm{\alpha}^\top \bm{m}' - \bm{\alpha}(e)\bm{m}'(e)) }{{(\bm{\alpha}^\top \bm{m}')}^2}
		-
		\frac{1}{2} \bm{m}'(e) \sum_{e' \in R: e' \neq e} \sigma_{e'} \frac{\bm{\alpha}(e')\bm{m}'(e')}{{(\bm{\alpha}^\top \bm{m}')}^2}.
		% \sum_i \sigma_i \left(\frac{a_i m_i}{\sum_j a_j m_j} - \frac{m_k}{\sum_j a_j m_j}\right) + \sigma_k\left(\frac{-m_k}{\sum_j a_j m_j}\right).
		\end{align*}
		The subdifferential includes zero if and only if
		\begin{align}
		\sum_{e' \in R: e' \neq e}\sigma_{e'}\bm{\alpha}(e')\bm{m}'(e') =
		\sigma_e (\bm{\alpha}^\top \bm{m}' -     \bm{\alpha}(e)\bm{m}'(e))
		= \sigma_e\sum_{e' \in R: e' \neq e} \bm{\alpha}(e') \bm{m}'(e').
		\label{eq:cond}
		\end{align}
		Noting that $\bm{\alpha}(e')\bm{m}'(e')$ are all non-negative, we have $\nabla_{\bm{\alpha}} h = \bm{0}$ only when all $\sigma_e\;(e \in R)$ have the same value.
		By some case analysis, we can show that this implies that $\bm{\alpha}$ is a multiple of the all-one vector, in which case $h$ takes the minimum value of $0$.
		Hence, $h$ takes the maximum value at boundary of ${[1-c,1]}^R$ and the claim follows.
	\end{proof}

	Now, we compute the maximum value of $h$.
	By the claim above, we only have to consider $\bm{\alpha} \in {\{1-c,1\}}^R$.
	Let $R_1 = \{e \in R \mid \bm{\alpha}(e) = 1\}$ and $R_{1-c} = \{e \in R \mid \bm{\alpha}(e) = 1-c\}$.
	Then, we have
	\begin{align*}
	h(\bm{\alpha}) &= \frac{c \cdot \sum_{e \in R_1 }\bm{m}'(e) \cdot \sum_{e \in R_{1-c}}  \bm{m}'(e)}{\sum_{e \in R}\bm{m}'(e) \cdot \left(\sum_{e \in R_1} \bm{m}'(e) +\sum_{e \in R_{1-c}} (1-c) \bm{m}'(e) \right)} \\
	& = \frac{cM_1 M_{1-c}}{(M_1+M_{1-c})(M_1 + (1-c)M_{1-c})},
	\end{align*}
	where $M_1 = \sum_{e \in R_1 }\bm{m}'(e)$ and $M_{1-c} = \sum_{e \in R_{1-c} }\bm{m}'(e)$.
	We can confirm
	\[
	h(\bm{\alpha}) \leq \frac{{\left(1-\sqrt{1-c}\right)}^2}{c},
	\]
	where the equality holds when $M_1 = \sqrt{1-c} \cdot M_{1-c}$.
	This proves the lemma.
\end{proof}
The following corollary bounds the earth mover's distance between $\mathcal{A}$ on $f$ and $f'$, given that $e^*$ is already selected in the process of $\mathcal{A}$ on $f$.
\begin{corollary}\label{cor:em-one-iteration}
	For any $S,T \subseteq E - e^*$ with $s := |S|=|T|-1 \leq k-1$, we have
	\[
	d_{\mathrm{EM}}\bigl(S+e^* \cup \mathcal{A}(f_{S+e^*}), T \cup \mathcal{A}(f'_T)\bigr) \leq |S \bigtriangleup T| + \frac{{\left(1-\sqrt{1-c}\right)}^2}{c} \cdot (k-s-1) + 1.
	\]
\end{corollary}
\begin{proof}
	We prove by (backward) induction on $s$.
	When $s = k-1$, the statement clearly holds because $\mathcal{A}(f_{S+e^*}) = \mathcal{A}(f'_T) = \emptyset$.
	
	Suppose that the statement holds for any set of size more than $s$.
	Let $S,T \subseteq E-e^*$ be sets of size $s$, and let $e_f = \mathcal{I}(f_{S+e^*})$ and $e_{f'} = \mathcal{I}(f'_{T})$.
	Then, there exists a distribution $\mu$ over $(E \setminus (S + e^*)) \times (E \setminus (T+e^*))$ such that the marginal distributions of the first and second coordinates are equal to the distributions of $e_f$ and $e_{f'}$, respectively, and
	\[
	d_{\mathrm{EM}}\bigl(S+e_f, T+e_{f'}\bigr)
	= \mathop{\mathbf{E}}_{(e_f,e_{f'}) \sim \mu}\bigl[|(S + e_f) \bigtriangleup (T + e_{f'})|\bigr],
	\]
	where we used the symbols $e_f$ and $e_{f'}$ in the RHS because the marginal distributions match.
	Then, we have
	\begin{align*}
	& d_{\mathrm{EM}}\bigl(\mathcal{A}(f_{S+e^*}), \mathcal{A}(f'_T)\bigr)  \\
	& \leq \sum_{e,e'}\Pr_{(e_f,e_{f'})\sim \mu}[e_f = e \wedge e_{f'} = e'] d_{\mathrm{EM}}\bigl(\mathcal{A}(f_{S+e+e^*}), \mathcal{A}(f'_{T+e'})\bigr) \\
	& = \sum_{e,e'} \Pr_{(e_f,e_{f'})\sim \mu}[e_f = e \wedge e_{f'} = e'] \left(|(S+e) \bigtriangleup (T+e')| + \frac{{\left(1-\sqrt{1-c}\right)}^2}{c}\cdot(k-s-2) + 1\right) \tag{by induction hypothesis} \\
	& = \mathop{\mathbf{E}}_{(e_f,e_{f'}) \sim \mu}\bigl[|(S + e_f) \bigtriangleup (T + e_{f'})|\bigr] + \frac{{\left(1-\sqrt{1-c}\right)}^2}{c}\cdot(k-s-2) + 1 \\
	& = |S\bigtriangleup T| + \frac{{\left(1-\sqrt{1-c}\right)}^2}{c}\cdot(k-s-1) + 1, \tag{by Lemma~\ref{lem:em-one-iteration}}
	\end{align*}
	as desired.
\end{proof}

For $1 \leq i \leq k$, let $e_i$ and $e'_i$ denote the element that $\mathcal{A}$ on $f$ and $f'$, respectively, selects at the $i$-th iteration.
Let $S_i = \{e_1,\ldots,e_i\}$ and $S'_i = \{e'_1,\ldots,e_i\}$.
Note that these are all random variables.
The following clearly holds.
\begin{proposition}\label{prop:same-distribution}
	For every $i \in [k]$ and $S \subseteq E - e^*$ of size $i-1$,
	We have
	\[
	\Pr[e_i = e \mid S_{i-1} = S \wedge e_i \neq e^*] = \Pr[e'_i = e \mid S'_{i-1} = S].
	\]
\end{proposition}
Now, we show the following, which implies Theorem~\ref{thm:ub-randomized-greedy} when $S=\emptyset$.
\begin{lemma}
	For any $S \subseteq E-e^*$, we have
	\[
	d_{\mathrm{EM}}\bigl(\mathcal{A}(f_S), \mathcal{A}(f'_S)\bigr) \leq \frac{{\left(1-\sqrt{1-c}\right)}^2}{c} \cdot (k-1) + 2.
	\]
\end{lemma}
\begin{proof}
	We prove by (backward) induction on the size of $S$.
	The claim clearly holds when $|S|=k$.
	
	Now, we assume that the claim holds for sets of size more than $s$.
	Let $S \subseteq E - e^*$ be a set of size $s$.
	By Proposition~\ref{prop:same-distribution}, there exists a distribution $\mu$ over $(E \setminus S) \times (E \setminus (S+e^*))$ such that the marginal distributions of $\mu$ on the first and second coordinates are equal to the distribution of $\mathcal{I}(f_S)$ and $\mathcal{I}(f'_S)$, respectively, and a pair $(e_f,e_{f'})$ sampled from $\mu$ always satisfies $e_f=e_{f'}$ if $e_f \neq e^*$.
	Then, we have
	\begin{align*}
	& d_{\mathrm{EM}}(\mathcal{A}(f_S),\mathcal{A}(f'_S))  \\
	& \leq \Pr_{(e_f,e_{f'})\sim \mu}[e_f = e^*] \cdot d_{\mathrm{EM}}\bigl(\mathcal{A}(f_{S+e^*}), \mathcal{A}(f'_{S+e_{f'}})\bigr) \\
	& \quad + \sum_{e \in E \setminus (S+e^*)}\Pr_{(e_f,e_{f'}) \sim \mu}[e_f = e \mid e_f \neq e^*] \cdot d_{\mathrm{EM}}\bigl(\mathcal{A}(f_{S+e}), \mathcal{A}(f'_{S+e})\bigr) \\
	& \leq \Pr_{(e_f,e_{f'})\sim \mu}[e_f = e^*]  \cdot \left(\frac{{\left(1-\sqrt{1-c}\right)}^2}{c} \cdot (k-s-1)+ 2\right) \\
	& \quad + \left(1-\Pr_{(e_f,e_{f'})\sim \mu}[e_f = e^*]\right) \left(\frac{{\left(1-\sqrt{1-c}\right)}^2}{c} \cdot (k-1) + 2\right) \tag{by Corollary~\ref{cor:em-one-iteration} and the induction hypothesis} \\
	& \leq \frac{{\left(1-\sqrt{1-c}\right)}^2}{c} \cdot (k-1) + 2. \qedhere
	% & \leq \sum_{i = s+1}^{k} \Pr[e_i = e^*] \max_{\substack{S \subseteq T \subseteq E-e^*:|T|=i-1,\\e' \in E \setminus (S + e^*)}}d_{\mathrm{EM}}\bigl(\{\mathcal{A}(f_{T+e^*})\}, \{\mathcal{A}(f'_{T+e'})\}\bigr) \\
	% & \leq \max_{s + 1 \leq i \leq k} \max_{\substack{S \subseteq T \subseteq E-e^*:|T|=i-1,\\e' \in E \setminus (S + e^*)}} d_{\mathrm{EM}}\bigl(\{\mathcal{A}(f_{T+e^*})\}, \{\mathcal{A}(f'_{T+e'})\}\bigr)
	% \leq \frac{{\left(1-\sqrt{1-c}\right)}^2}{c}\cdot (k-s)+1. \tag{by Corollary~\ref{cor:em-one-iteration}}
	\end{align*}
\end{proof}

\subsection{Lower Bound for the Randomized Greedy Algorithm}
In this section, we show the following lower bound for the randomized greedy algorithm, which shows that the bound in Theorem~\ref{thm:ub-randomized-greedy} is tight.
% 	To obtain a lower bound, we realised that what we need is a function with many modular elements, and minimal elements that affect other marginal returns. This gives us our first result in the worst-case sensitivity problem:
\begin{theorem}\label{thm:lb-randomized-greedy}
	For any $c \in (0,1]$ and $\varepsilon > 0$, the sensitivity of the randomized greedy algorithm on monotone submodular functions with curvature $c$ is at least $\frac{{1-(1-\sqrt{c})}^2}{c}-\varepsilon$.
\end{theorem}

% Figure~\ref{fig:comparec} shows the coefficients of the lower and upper bounds given in Theorems~\ref{thm:lb-randomized-greedy} and~\ref{thm:ub-randomized-greedy}, respectively, as a function of $c$.
% We can observe that as long as $c$ is not very close to $1$, our upper bound is quite acceptable. \ynote{I'm not sure we want provide such a subjective claim.}
% \begin{align}
% c_{up} &= \frac{c}{{(1 + \sqrt{1-c})}^2} \\ \nonumber
% c_{low} &= \frac{c}{4-2c} \nonumber
% \end{align}
% \begin{figure}
% 	\centering
% 	\includegraphics[width=.5\hsize]{cComparison.eps}
% 	\caption{A comparison of the quantities $c_{up}$ and $c_{low}$, as $c$ varies from $0$ to $1$.}\label{fig:comparec}
% \end{figure}
% A similar result to Theorem~\ref{uppercurvresult} can also be stated from rearranging $c_{\mathrm{low}}$
% We see that the bound is tight up to a constant for $c \geq 4k/{(k+1)}^2$ and for $c \leq 4/(k+2)$.

Let $n$ be an integer, $c \in (0,1]$, and $\alpha = (1-\sqrt{1-c})/c$.
Let $A,B$ be sets of size $(1-\alpha)n$ and $\alpha n$, respectively, where we assume $\alpha n$ is an integer, and let $E = A \cup B \cup \{e^*\}$ be a set of size $n + 1$, where $e^*$ is a newly introduced element.
Now, we define a function $f\colon 2^E \to \mathbb{R}_+$ as
\[
f(S) = M x_{e^*} + \sum_{e \in A} x_e + \sum_{e \in B} \Bigl((1-x_{e^*})x_e + (1-c)x_{e^*}x_e\Bigr),
\]
where $M \in \mathbb{R}$ is a large value determined later, and $x_e\;(e \in E)$ is the indicator of the event that $e$ belongs to $S$.
The marginal gain of an element $e^*$ is $M$, which never changes.
All other marginal gains are initially $1$, and those of the elements in $A$ never change whereas those of the elements in $B$ decreases to $1-c$ after $e^*$ is selected.
We can confirm that $f$ is a monotone submodular function with curvature $c$.
Our goal is to show that $d_{\mathrm{EM}}(\mathcal{A}(f),\mathcal{A}(f'))$ is large, where $f' = f^{\setminus e^*}$.

% \begin{lemma}\label{lem:lb-randomized-greedy-A}
%     For every $e \in A$, we have
%           \[
%               d_{\mathrm{EM}}(\mathcal{A}(f),\mathcal{A}(f^{\setminus e})) = O\left(\frac{k}{n}\right).
%           \]
% \end{lemma}
% \begin{lemma}\label{lem:lb-randomized-greedy-B}
%     For every $e \in B$, we have
%           \[                d_{\mathrm{EM}}\left(\mathcal{A}(f),\mathcal{A}(f^{\setminus e})\right) = O\left(\frac{k}{n}\right).
%           \]
% \end{lemma}

% The following lemma is useful, which states that, whenever $e^*$ is available, the randomized greedy algorithm selects $e^*$ in the first iteration with high probability.
% \begin{lemma}\label{lem:lb-randomized-greedy-first-step}
%     By setting $M = \Omega(n)$, we have
%     \[
%         \Pr[\mathcal{I}(f) \neq e^*] = O\left(\frac{1}{n}\right) \quad \text{and} \quad \Pr[\mathcal{I}(f^{\setminus e}) \neq e^*] = O\left(\frac{1}{n}\right).
%     \]
%     for every $e \in E-e^*$.
% \end{lemma}
% \begin{proof}
%     We have
%     \[
%         \Pr[\mathcal{I}(f) \neq e^*] = \frac{1}{M + 2m} = O\left(\frac{1}{n}\right).
%     \]
%     from the choice of $M$.
%     We can show $\Pr[\mathcal{I}(f^{\setminus e}) \neq e^*] = O(1/n)$ similarly.
% \end{proof}

% \subsubsection{Proof of Lemma~\ref{lem:lb-randomized-greedy-e1}}
% 		We will show we instead remove $e_1$ by demonstrating that if we do so, we get average sensitivity in the Theorem's statement, which approaches $\mathcal{O}(k)$ when $n$ is large.

For $e \in A$ (resp., $e \in B$), let $p_A$ (resp., $p_B$) be the probability that the output of $\mathcal{A}(f_{e^*})$ includes $e$.
Note that the choice of $e$ does not matter.
The following two lemmas bound $p_A$ and $p_B$ from below and above, respectively.
We let $D = (1-\alpha)n + (1-c)\alpha n$.
\begin{lemma}\label{lem:bound-on-pA}
	We have
	\[
	p_A \geq \frac{k}{D} \left(1 - \frac{k-1}{2(D - k)}\right).
	\]
\end{lemma}
\begin{proof}
	For $i \in [k]$, let $p_i$ be the probability that $e$ is selected at the $i$-th iteration, given that $e$ has not been selected in the first $i-1$ iterations.
	We note that $1/D \leq p_i \leq 1/(D-k)$ for any $i \in [k]$ because the marginal gain of adding an element in $A$ is $1$ whereas the sum of marginal gains is between $D$ and $D-k$ throughout the algorithm.
	Then, we have
	\begin{align*}
	p_A & = \sum_{i=1}^k \prod_{j=1}^{i-1}(1-p_j)p_i
	\geq \sum_{i=1}^k {\left(1-\frac{1}{D-k}\right)}^{i-1}\frac{1}{D}
	= \frac{D - k}{D} \left(1 - {\left(1 - \frac{1}{D - k}\right)}^k\right)\\
	& \geq \frac{D - k}{D} \left(1 - {\left(1 - \frac{k}{D - k} + \binom{k}{2}\frac{1}{{(D - k)}^2}\right)}\right)
	= \frac{k}{D} \left(1 - \frac{k-1}{2(D - k)}\right). \qedhere
	\end{align*}
\end{proof}

\begin{lemma}\label{lem:bound-on-pB}
	We have
	\[
	p_B \leq \frac{(1 - c)k}{D-k}.
	\]
\end{lemma}
\begin{proof}
	For $i \in [k]$, let $p_i$ be the probability that $e$ is selected at the $i$-th iteration, given that $e$ has not been selected in the first $i-1$ iterations.
	We note that $(1-c)/D \leq p_i \leq (1-c)/(D-k)$ for any $i \in [k]$ because the marginal gain of adding an element in $B$ is $1 -c$ (note that we have already selected $e^*$) whereas the sum of marginal gains is between $D$ and $D-k$ throughout the algorithm.
	Then, we have
	\begin{align*}
	p_B & = \sum_{i=1}^k \prod_{j=1}^{i-1}(1-p_j)p_i
	\leq \sum_{i=1}^k {\left(1-\frac{1-c}{D}\right)}^{i-1}\frac{1-c}{D-k} \\
	& = \frac{D}{D-k} \left(1 - {\left(1 - \frac{1 - c}{D}\right)}^k\right)
	\leq \frac{(1 - c)k}{D-k}. \qedhere
	\end{align*}
\end{proof}

\begin{lemma}\label{lem:lb-after-selecting-e^*}
	For any $\varepsilon > 0$, there exists $n$ such that
	\[
	d_{\mathrm{EM}}\left(\mathcal{A}(f_{e^*}), \mathcal{A}(f')\right) >  \frac{{(1 - \sqrt{1-c})}^2}{c} \cdot k - \varepsilon.
	\]
\end{lemma}
\begin{proof}
	We first note that, for any $e \in A \cup B$, the probability that $\mathcal{A}(f')$ selects $e$ is $p := k/n$, and that $p_A \geq p \geq p_B$.
	Then, we have
	\begin{align*}
	& d_{\mathrm{EM}}\left(\mathcal{A}(f_{e^*}), \mathcal{A}(f')\right)
	= \frac{1}{2}\sum_{e \in A}\left|p_A - \frac{k}{n}\right| + \frac{1}{2}\sum_{e \in B}\left|p_B - \frac{k}{n}\right| \\
	& = \frac{1}{2}\sum_{e \in A}\left(p_A - \frac{k}{n}\right) + \frac{1}{2}\sum_{e \in B}\left(\frac{k}{n} - p_B\right) \\
	& \geq
	\frac{1}{2}\sum_{e \in A}\left(\frac{k}{D} \left(1 - \frac{k-1}{2(D - k)}\right) - \frac{k}{n}\right) + \frac{1}{2}\sum_{e \in B}\left(\frac{k}{n} - \frac{(1 - c)k}{D-k}\right)  \tag{by Lemmas~\ref{lem:bound-on-pA} and~\ref{lem:bound-on-pB}} \\
	& = \frac{k}{4} \left(\frac{3 k (1 - \alpha) + \alpha +  (4 - 2c) D \alpha-2D-1}{(k - D) D} \cdot n + 4 \alpha - 2\right) \\
	& = \frac{k}{4} \frac{k (4 c \alpha^2 - (1 + 2 c) \alpha - 1) + (1 - \alpha) (1 + 4 c n \alpha - 4 c^2 n \alpha^2)}{(1 - c \alpha) (n (1 - c \alpha) - k)}.
	\end{align*}
	which converges (as $n \to \infty$) to
	\begin{align*}
	\frac{c (1 - \alpha) \alpha}{1 - c \alpha} \cdot k = \frac{{(1 - \sqrt{1-c})}^2}{c} \cdot k
	\end{align*}
	and hence the claim holds.
\end{proof}

\begin{proof}[Proof of Theorem~\ref{thm:lb-randomized-greedy}]
	We have
	\begin{align}
	\Pr[\mathcal{I}(f) \neq e^*] = \frac{1}{M + n} = O\left(\frac{1}{n}\right). \label{eq:whp-e^*}
	\end{align}
	by choosing $M = \Omega(n^2)$.
	
	Let $\mu$ be the probability distribution over $E \times (E-e^*)$ such that the marginal distributions of $\mu$ on the first and second coordinates are equal to the distributions of $\mathcal{I}(f)$ and $\mathcal{I}(f')$, and
	\[
	d_{\mathrm{EM}}(\mathcal{A}(f),\mathcal{A}(f'))
	= \sum_{e \in E,e' \in E-e^*}\Pr_{(e_f,e_{f'}) \sim \mu}[e_f = e \wedge e_{f'} = e'] d_{\mathrm{EM}}\bigl(\mathcal{A}(f_e),\mathcal{A}(f'_{e'})\bigr).
	\]
	Let $\mu_{e^*}$ be the distribution $\mu$ conditioned on the sampled element not being $e^*$, and let $\tilde{\mu}_{e^*}$ be the distribution closest to $\mu_{e^*}$ in total variation distance such that the marginal distributions of $\mu_{e^*}$ and $\tilde{\mu}_{e^*}$ on the first coordinate are the same and the marginal distribution of $\tilde{\mu}_{e^*}$ on the second coordinate is equal to $\mathcal{I}(f')$.
	Then, we have $d_{\mathrm{TV}}(\mu_{e^*},\tilde{\mu}_{e^*}) = O(1/n)$ because $d_{\mathrm{TV}}(\mu,\mu_{e^*}) = O(1/n)$ by~\eqref{eq:whp-e^*}.
	% \ynote{maybe this is sloppy.}
	
	Then, we have
	\begin{align*}
	& d_{\mathrm{EM}}\bigl(\mathcal{A}(f),\mathcal{A}(f')\bigr) \\
	& \geq \sum_{e,e'}\Pr_{(e_f,e_{f'}) \sim \mu_{e^*}}[e_f = e \wedge e_{f'} = e'] d_{\mathrm{EM}}\bigl(\mathcal{A}(f_e),\mathcal{A}(f'_{e'})\bigr) - k \cdot d_{\mathrm{TV}}(\mu,\mu_{e^*})\\
	& \geq \sum_{e,e'}\Pr_{(e_f,e_{f'}) \sim \tilde{\mu}_{e^*}}[e_f = e \wedge e_{f'} = e'] d_{\mathrm{EM}}\bigl(\mathcal{A}(f_e),\mathcal{A}(f'_{e'})\bigr) - k \cdot \bigl(d_{\mathrm{TV}}(\mu,\mu_{e^*}) + d_{\mathrm{TV}}(\mu_{e^*},\tilde{\mu}_{e^*})\bigr)\\
	& \geq d_{\mathrm{EM}}(\mathcal{A}(f_{e^*}),\mathcal{A}(f')) - O\left(\frac{k}{n}\right) \\
	& \geq  \frac{{(1 - \sqrt{1-c})}^2}{c} \cdot k - \varepsilon, \tag{by Lemma~\ref{lem:lb-after-selecting-e^*}}
	\end{align*}
	by choosing $n$ to be large enough.
\end{proof}

\section{Distributed Algorithms}\label{appendix:distributed}
In this section, we apply the techniques discussed so far to the theory of distributed algorithms. In particular, we aim to show that for any algorithm that we have shown to attain $\Omega(k)$ sensitivity, its corresponding distributed algorithm has $\Omega(k)$ sensitivity. Here the correspondence is given in~\cite{barbosa2016new}.

For our model of distributed computation, we will follow the example of~\cite{barbosa2016new} and use the Massively Parallel Communication (MPC) model. Particular details we note for our purposes are for this model there are $M$ machines, each with space $S$, and we have $MS = \mathcal{O}(n)$, where $n$ is the size of our ground set. Additionally, following the example of~\cite{barbosa2016new} again, we will restrict both $M, S$ to be $< n^{1 - \Omega(1)}$. Note this means that both $M, S$ will be $\omega_n(1)$.

\begin{algorithm}[]
	\caption{Randomized Distributed Greedy}\label{alg:1stdistgreedy}
	\Procedure{\emph{\Call{GreeDi}{$f,k,m$}}}{
		\Input{Monotone submodular function $f\colon 2^E \to \mathbb{R}_+$, an integer $k$, a number of machines $m$.}
		\For{each $e \in E$}{
			Assign $e$ to a random machine $i \in [m]$\;
		}
		Let $V_i$ be the elements on machine $i$\;
		\For{each $i \in [m]$}{
			Let $S_i$ be the output of the deterministic greedy on $V_i$\;
		}
		Let $S = \bigcup_{i \in [m]} S_i$, and place it on machine $1$\;
		Let $S'$ be the output of the deterministic greedy on $S$\;
		Let $T = \arg \max\{f(S'), f(S_i)\}$\;
		\Return $T$\;
	}
\end{algorithm}

\subsection{Warm-up: Simple Greedy}
We begin with a simpler version of a distributed greedy than that given in~\cite{barbosa2016new}, which is found in~\cite{barbosa2015power}. Similar to the authors we will call this algorithm \Call{GreeDi}{}, and its definition is given in Algorithm~\ref{alg:1stdistgreedy}.

The function we will use is:
\begin{align}
f(S) = Cx_1 + (1-cx_1)\sum_{i=2}^{n/2}x_i + \left(1-\frac{c}{2}\right)\sum_{i=n/2+1}^n x_i \label{eqn::GreeDiTest}
\end{align}
where $x_i = 1[e_i \in S]$, $C$ is a large constant, and $c$ is the curvature.
\begin{theorem}\label{theorem:GreeDi}
	The algorithm \Call{GreeDi}{} ran on the function in Equation~\eqref{eqn::GreeDiTest} attains $\Omega(k)$ sensitivity for sufficiently large $n$, for all $0 < c \leq 1$.
\end{theorem}
Due to the prior analysis of the greedy algorithm in a non-distributed setting, and noting that the optimal set consists of $e_1$ and $k-1$ elements from $\{e_{n/2+1},\ldots, e_n\}$, we see that instead proving the following lemma will instantly give our result, noting that $e_1$ will always feature in the final solution and be chosen first:
\begin{lemma}\label{lem:GreeDiLem}
	Consider \Call{GreeDi}{} ran on the function in Equation~\eqref{eqn::GreeDiTest}.
	The set $S$ in the algorithm contains at least $k-1$ elements from each of the subsets $\{e_2, \ldots, e_{n/2}\}$ and $\{e_{n/2+1}, \ldots, e_n\}$ with probability $1-o_n(1)$.
\end{lemma}
\begin{proof}[Proof of Lemma~\ref{lem:GreeDiLem}]
	Throughout this proof, we fix $k$. Let the machine $V_i$ be the machine that $e_1$ is put on. Note that as we have $M < n^{1 - \Omega(1)}$, then for large enough $n$, we have that at least $k$ elements from the subset $\{e_{n/2+1}, \ldots, e_n\}$ must also be included in $V_i$.
	More specifically, let $M < n^{c'}$, with $c' < 1$. Note that the distribution of how many elements from the subset $\{e_{n/2+1}, \ldots, e_n\}$ feature on $V_i$ is a Bernoulli trial, with $p < n^{c'}$.
	Thus, we can bound the probability of choosing less than $k$ of these elements on $V_i$ by
	\begin{equation}
	\sum_{i=0}^k \binom{n^{c'}}{i} {\left(\frac{1}{n^{c'}}\right)}^i{\left(1-\frac{1}{n^{c'}}\right)}^{n-i}.
	\end{equation}
	We then note that this quantity converges to $0$ as $n \rightarrow \infty$, thus showing that $V_i$ must be assigned at least $k-1$ elements from $\{e_{n/2+1}, \ldots, e_n\}$. As $e_1$ will be selected by this machine, we see that its output will also contain $k-1$ elements from $\{e_{n/2+1}, \ldots, e_n\}$.
	From the definition of the function, it is clear that any other machine $V_j\; (j \neq i)$ will select as many elements from the subset $\{e_{2}, \ldots, e_{n/2}\}$ as it can, thus giving our conclusion.
\end{proof}
Lemma~\ref{lem:GreeDiLem} and the previous analysis on proportional algorithms give Theorem~\ref{theorem:GreeDi}.

\subsection{General Distributed Algorithm}
In this subsection we now tackle the general framework for distributed algorithms given in~\cite{barbosa2016new}. We use the same model for distributed computation, and restrict our attention to the setting of Section~\ref{subsec:GenHard}, where the algorithm was defined by $k$ distributions over positive integers $D_{k,1}, \ldots, D_{k,k}$. For any algorithm $\mathcal{A}$ that we have used in a non-distributed setting, we call $\mathcal{A}_{d}$ the corresponding distributed algorithm formed according to the argument in~\cite{barbosa2016new}.

\subsubsection{Algorithm Framework}
We first take the time to explain the algorithm framework for generating a distributed algorithm for a centralized algorithm explained in~\cite{barbosa2016new}.

Firstly, we take a centralized sequential algorithm $\mathcal{A}$ with approximation guarantee $\alpha$, and let $\varepsilon > 0$ be the desired accuracy. Namely, our distributed algorithm will have an approximation factor $\alpha - \varepsilon$.

We have a total of $gm$ machines, divided into $g$ groups with $g = \Theta(1/(\alpha \varepsilon))$. The number of machines $m$ is chosen arbitrarily, in accordance with the restrictions we already laid out when discussing our general MPC framework.

The algorithm consists of $\Theta(1/\varepsilon)$ runs, where at each run we maintain an incumbent best solution $S^*$, and a pool of elements $C_{r-1}$ for each run $r$, with $C_0 = \emptyset$. In a run $r$, we take each group of machines $g$ and distribute the ground set $V$ uniformly on them, as in \Call{GreeDi}{}. This is done for each group of machines, and then the algorithm $\mathcal{A}$ is ran on each machine, taking the ground set to be the elements chosen in the distribution, and the pool $C_{r-1}$. Each machine $i$ will return a set $S_{i,r}$.

After that is completed, we update $S^*$ if applicable, by taking the max of $S_{i,r}$ over all $i$ and the current best solution $S^*$, and update the pool of elements $C_{r-1}$ to include the union of all solutions $S_{i,r}$ constructed during run $r$ to give $C_r$. After all runs, the algorithm returns $S^*$.

\subsubsection{Sensitivity Analysis}
In this setting we have $g$ groups of $m$ machines. We take $g = \Theta(1/(\alpha\varepsilon))$, where $\alpha$ is the approximation factor of the corresponding centralized algorithm, and $\varepsilon$ is the additive error to that factor. We assume that $\varepsilon$ does not depend on $n$. Note that this means that the number of machines in each group is $\omega_n(1)$.

We prove the following:
\begin{theorem}
	Let $\mathcal{A}$ be some algorithm with $\Omega(k)$ sensitivity, as shown by Theorem~\ref{mainhard}. Then $\mathcal{A}_d$ also has $\Omega(k)$ sensitivity for some function.
\end{theorem}
\begin{proof}
	Consider the function:
	\begin{equation}
	f(S) = C\sum_{i=1}^k x_i + (1-cx_j)\sum_{i=k+1}^{n/2} x_i + \left(1-\frac{c}{2}\right)\sum_{n/2+1}^{n} x_i
	\end{equation}
	where $x_i = 1[e_i \in S]$, $e_j$ is an element chosen similarly to Theorem~\ref{mainhard}, and $C$ is a very large constant. This function roughly serves the purpose of the large-element scheme in the centralized case.
	
	We begin by considering the first run. Note that for very large $n$ and fixed $k$, we have that with probability $1 - o_n(1)$ that within each group $V$ of machines, that each machine $V_i$ will be assigned at most one of the elements $\{e_1, \ldots, e_k\}$. Due to the fact that our algorithm is a constant-factor approximation algorithm, that element must be selected in each machine. Therefore, we see that all of $\{e_1, \ldots, e_k\}$ must feature in a solution, and so we have that $C_1$ will contain all of those elements. %$\{e_1, \ldots, e_k\}$.
	
	Additionally, due to the argument of Lemma~\ref{lem:GreeDiLem}, we see that within each group of machines $V$, that the machine $V_i$ which is assigned $e_j$ must also be assigned at least $k-1$ elements from the set $\{e_{n/2+1}, \ldots, e_n\}$ with probability $1 - o_n(1)$. Therefore, we see that $C_1$ will contain $k-1$ elements from $\{e_{n/2+1}, \ldots, e_n\}$. Similarly, it must also contain $k-1$ elements from $\{e_{k+1}, \ldots, e_{n/2}\}$ with probability $1 - o_n(1)$, by the logic used in the previous Theorem.
	
	As the sets $C_r$ grow larger with $r$, we see that on the final run, all machines in all groups will run the greedy algorithm on a set which contains all of the large elements $\{e_1, \ldots, e_k\}$, and at least $k-1$ of both $\{e_{k+1}, \ldots, e_{n/2}\}$ and $\{e_{n/2+1}, \ldots, e_n\}$ with probability $1 - o_n(1)$.
	
	The prior analysis of the algorithm $\mathcal{A}$ thus gives our conclusion. We note that from earlier argument, from run $2$ onwards the algorithm must behave as in the centralized case for this function, excluding the possibility of a higher objective value, as $e_j$ will always be present in the correct position.
\end{proof}

\section{Average Sensitivity}\label{appendix:AveAll}
Throughout the course of this document, we have presented all of our results in the context of worst-case sensitivity. We can define average sensitivity similar to Definition \ref{defn:sensRand} but instead of taking a maximum over deleting any element, we take an expectation over deleting one randomly:
\begin{definition}
	Let $\mathcal{A}$ be a randomized algorithm that takes a function $f\colon 2^E \to \mathbb{R}$ and returns a set $S \subseteq E$.
	The \emph{average sensitivity} of $\mathcal{A}$ on $f$ is
	\[
	\mathbb{E}_{e \in E} \left\{d_{\mathrm{EM}}(\mathcal{A}(f), \mathcal{A}(f^{\setminus  e}))\right\},
	\]
	where we identified $\mathcal{A}(f)$ and $\mathcal{A}(f^{\setminus e})$ with the distributions of $\mathcal{A}(f)$ and $\mathcal{A}(f^{\setminus e})$, respectively.
\end{definition}

In this section, we extend all the results to average sensitivity using a simple technique. Many of the worst-case results relied on constructing a function that was nearly modular, apart from one element $e_j$ which when selected, greatly changed the order of marginals, giving a factor of $(1-ce_j)$ in front of many of the marginals. To adapt our results, we will change that factor of $1-ce_j$ to $1-e\mathcal{I}$, where $\mathcal{I}$ corresponds to the indicator variable when a subset of elements of size proportional to $k$ is selected.

From a high level, we will see that when any element is deleted in $\mathcal{I}$, the function will always be modular, and this will give us our sensitivity in a similar way to the worst-case setting. In the worst-case sensitivity setting the trivial upper bound on sensitivity was $\mathcal{O}(k)$, but here it will be $\mathcal{O}(k^2/n)$.

%\cnote{Only Theorem~\ref{theorem:AveGeneral} to finish}

\subsection{General case}\label{appendix:PropAve}
\begin{theorem}\label{theorem:ProportionalAve}
	For proportional algorithms, some function $f$ attains $\Omega(k^2/n)$ sensitivity.
\end{theorem}
\begin{proof}
	We assume in this proof that $k$ is even, for the case where $k$ is odd we can change instances of $k/2$ to $(k+1)/2$. Let the indicator $\mathcal{I}$ take the value $1$ if all elements $e_1, \ldots, e_{k/2}$ are selected. Then consider the function:
	\begin{equation}
	f(e_1, \ldots, e_n) = \sum_{i=1}^{k/2} A_i e_i + (1-\mathcal{I})\sum_{i=k+2}^n B_i e_i + \sum_{i=k/2+1}^{k+1} C_i e_i
	\end{equation}
	and the constants $A_i, B_i, C$ are such that $A_1 \gg \ldots \gg A_{k/2} \gg B_{k+2} \gg \ldots \gg B_n \gg C_{k/2+1} \gg \ldots \gg C_{k+1}$.
	In particular, for small enough $\delta = o(1/n)$ and $\varepsilon = \varepsilon_\mathcal{R}(\delta)$, we want the following:
	\begin{align*}
	\frac{A_i}{\sum_{j=i+1}^{k+2}A_j+\sum_{j=k+2}^n B_i+\sum_{k/2+1}^{k+1}C_i} &> 1-\varepsilon, \\
	\frac{B_i}{\sum_{j=i+1}^n B_j+\sum_{k/2+1}^{k+1}C_i} &> 1-\varepsilon, \\
	\frac{C_i}{\sum_{j=i+1}C_j} &> 1-\varepsilon,
	\end{align*}
	where the first equation holds for all $i \in \{1, \ldots, k/2\}$, the second holds for all $i \in \{k+2, \ldots, n\}$, and the third holds for all $i \in \{k/2+1, \ldots, k+1\}$. It is clear that $f$ is monotone submodular.
	
	We now analyze the output of the sequential algorithm on $f$ and $f^{\setminus e_i}$ for $i \in [n]$.
	
	\begin{claim}
		We have
		\[
		\Pr\left[|\mathcal{A}(f,k,\mathcal{R}) \cap \{e_1,\ldots,e_{k+1}\}| = k\right] > 1 - (k/2)\delta.
		\]
	\end{claim}
	\begin{proof}
		From the choice of $A, B, C$, we first choose the sequence $e_1, \ldots, e_{k/2}$ with probability at least $1-\delta$ at each step. An application of Bernoulli's inequality shows this happens with probability at least $1-(k/2)\delta$.
		When this happens, the marginal gains of $e_i$ for $i > k+1$ are $0$ in subsequent steps, and hence we choose only elements $e_i$ for $i \leq n/2$.
	\end{proof}
	\begin{claim}
		We have
		\[
		\Pr\left[|\mathcal{A}(f^{\setminus e_i}, k,\mathcal{R}) \cap \{e_{k+2}, \ldots, e_n\}| = k/2+1 \right] > 1-k\delta , i \in{1, \ldots, k/2}
		\]
	\end{claim}
	\begin{proof}
		As in the previous claim, we see that we first choose the elements $\{e_1, \ldots, e_{k/2}\} \setminus e_i$ with probability $(1-\delta)^{k/2-1}$.
		
		By the choice of $B_i, C_i$, and the fact that $\mathcal{I}$ will still take the value $0$, we see that we successively choose the elements $e_{k+2}, \ldots, e_{3k/2+3}$ with probability at least $1-\delta$ in each case.
		
		We then see that the probability of selecting all of these elements in turn is at least $(1-\delta)^k$, and using Bernoulli's inequality gives our result.
	\end{proof}
	\begin{claim}
		We have
		\[
		\Pr\left[|\mathcal{A}(f^{\setminus e_i}, k,\mathcal{R}) \cap \{e_{k+2}, \ldots, e_n\}| \leq 1 \right] > 1-k\delta , i > k/2
		\]
	\end{claim}
	\begin{proof}
		Similar to the first claim in this proof, we see that we will first select the elements $e_1, \ldots, e_{k/2}$ with high probability. If $i > k+1$, then we select $\{e_{k/2+1},\ldots,e_{k+1}\}$ with high probability, and the probability of this happening will be at least $1 - k\delta$ by earlier argument.
		
		Instead, if $k/2 < i \leq k+1$, then after selecting $e_1, \ldots, e_{k/2}$, we will then select $\{e_{k/2+1},\ldots,e_{k+1}\} \setminus e_i$ and then $e_{k+2}$. This all happens with high probability, and again by similar argument it will be at least $1 - k\delta$, thus giving our conclusion.
	\end{proof}
	The claims above immediately show the sensitivity of $\mathcal{A}(\cdot,k,\mathcal{R})$ is $\Omega(k^2/n)$
	%\cnote{Change the C to $C_i$}
\end{proof}

\begin{theorem}\label{theorem:RandomizedAve}
	The Algorithm \Call{RandomizedGreedy}{} from Algorithm~\ref{alg:greedy} attains $\Omega(k^2/n)$ average sensitivity for some function $f$.
\end{theorem}
\begin{proof}
	We assume in this proof that $k$ is even, for the case where $k$ is odd we can change instances of $k/2$ to $(k+1)/2$. Let the indicator $\mathcal{I}$ take the value $1$ if all elements $e_1, \ldots, e_{k/2}$ are selected. The function we use is $f: 2^E \rightarrow \mathbb{R}$ with:
	\begin{equation}
	f(S) = C\sum_{i=1}^{k/2}x_i + (1-\mathcal{I})\sum_{i=k/2+1}^{2k+1}x_i + 0.5\sum_{i = 2k+2}^{n}x_i
	\text{ where }x_i = 1[e_i \in S]\text{ for each }i \in [n].
	\end{equation}
	where $C$ is large enough so the function remains monotone submodular. We now analyze the output of this algorithm on $f$ and $f^{\setminus e_i}$.
	\begin{claim}\label{claim:BuchbinderAve1}
		For $p_i = {\left(\frac{k-1}{k}\right)}^i$, we have
		\begin{equation}
		\Pr\left[|\mathcal{A}(f, k,\mathcal{R}) \cap \{e_{2k+2}, \ldots, e_n\}| = i \right] > p_{k-i-1}/k. \nonumber
		\end{equation}
	\end{claim}
	\begin{proof}
		We see that $p_i$ is equal to the probability that $e_1$ is not selected after $i$ iterations. Additionally, we see after we get $\mathcal{I}=1$, that all further elements under consideration will be from the set $e_{2k+2}, \ldots, e_n$.
		
		Noting that we require $e_1$ to be selected to get $\mathcal{I} = 1$, we see that we can bound the required probability below by $p_{k-i-1}/k$, such as in the worst-case sensitivity proof.
	\end{proof}
	\begin{claim}
		For $p_i = {\left(\frac{k-1}{k}\right)}^i$, we have
		\begin{equation}
		\Pr\left[|\mathcal{A}(f^{\setminus e_i}, k,\mathcal{R}) \cap \{e_{2k+2}, \ldots, e_n\}| = i \right] > p_{k-i-1}/k, i > k/2 \nonumber
		\end{equation}
	\end{claim}
	\begin{proof}
		In this case we see that no large elements are deleted, so the indicator variable $\mathcal{I}$ can still be set to $1$. The proof now proceeds as in the previous claim.
	\end{proof}
	
	\begin{claim}\label{claim:BuchbinderAve3}
		We have
		\[
		|\mathcal{A}(f^{\setminus e_i},k,\mathcal{R}) \cap \{e_1,\ldots,e_{2k+1}\} \setminus e_i| = k, i \in \{1, \ldots, k/2\}
		\]
	\end{claim}
	\begin{proof}
		As we select from the top $k$ marginals, at the first step we will select one of $e_1, \ldots, e_{k+1}$ (excluding $e_i$), and then add $e_{k+2}$ to the set of elements that we could possibly choose. Note that as the indicator variable $\mathcal{I}$ will never take the value $1$, as we have deleted one of the elements necessary.
		
		After $k$ steps, the element with maximum index we could include is $e_{2k+1}$, giving our claim.
	\end{proof}
	
	These claims enable us to give a simple bound on the sensitivity of this function:
	\begin{align*}
	&d_{\mathrm{EM}}(\mathcal{A}(f), \mathcal{A}(f^{\setminus  e})) \nonumber \\
	&\geq \sum_{i=1}^{k/2}d_{\mathrm{EM}}(\mathcal{A}(f), \mathcal{A}(f^{\setminus  e_i})) \nonumber \\
	&\geq \frac{k}{2n}\sum_{j=0}^{k-1} p_{k-j-1} \frac{2k-2j}{k} \\
	&=
	\frac{k}{2n}2k \left(1 - 2 {\left(\frac{k-1}{k}\right)}^k\right)
	= \Omega(k^2/n)
	\end{align*}
	giving us our result. The first inequality follows from only considering when elements with index $k/2$ are deleted as a simple upper bound, and the second inequality follows from the first and third Claims showing we can treat each summand the same in terms of sensitivity.
	
\end{proof}

\subsection{Bounded curvature}\label{appendix:AveGeneral}
\begin{theorem}\label{theorem:AveGeneral}
	Let $\mathcal{A}$ be a sequential algorithm that satisfies the hypotheses of Theorem~\ref{mainhard}. Then there is a function that attains $\mathcal{O}(k^2/n)$ sensitivity, even for curvature equal to $0 < c \leq 1$.
\end{theorem}

For this Theorem, we can use a variant of Lemma~\ref{lem:hardmain}. For this Lemma, we extend the notation $P_{i^*}$ to take a set as its argument. We define $P_{i^*}(S)$ to be the probability that all elements in $S$ are selected by step $i^*$. We can now state the Lemma as follows:
\begin{lemma}\label{lem:hardmainAve}
	Let $\mathcal{A}$ be an algorithm, $f:2^E \to \mathbb{R}_+$ be a monotone submodular function, $T \subset E$ be a collection of elements of size proportional to $k$.
	Suppose that we have $P_{i^*}(T)$ is $\Omega(1)$ for some $i^* \in [k]$ with $k-i^* = \Omega(k)$.
	Additionally, suppose there exist subsets of elements $E_1 = \{e_1, \ldots, e_{a-1}\}$ and $E_0 = \{e_a,e_{a+1},\ldots,e_b\}$ with $T \subset E_1$ such that $\mathcal{A}(f^{\setminus e})$, with $e \in T$, only selects elements from $E_1$, and $\mathcal{A}(f)$ only selects from $E_1$ before all elements in $T$ have been chosen, and only selects elements from $E_0$ after all elements in $T$ have been chosen.
	Then the function $f$ attains $\Omega(k)$ sensitivity for the algorithm $\mathcal{A}$.
\end{lemma}
\begin{proof}
	Note that the sensitivity can be bounded from below by the following total variation distance between the probabilities of selecting elements $e_i$ in $\mathcal{A}(f)$ and $\mathcal{A}(f^{\setminus e})$, with $e \in T$. Fix a particular $e \in T$ and we have:
	\begin{align*}
	d_{\mathrm{EM}}(\mathcal{A}(f), \mathcal{A}(f^{\setminus e})) \geq \sum_{i=1}^n |P_k(e_i) - P_k^{\setminus e}(e_i)|
	\end{align*}
	we can again bound this from below by restricting the summation to elements in the set $E_0$. Then noting that $\Omega(k)$ steps are still taken when we start selecting from $E_0$, and that $P_{i^*}(T)$ is $\Omega(1)$, we have:
	%and letting $\mathcal{I}^*$ be the event of $e^*$ being selected at step $i^*$:
	\begin{align*}
	\sum_{i=1}^n |P_k(e_i) - P_k^{\setminus e}(e_i)| \geq \sum_{e' \in E_0} |P_k(e) - P_k^{\setminus e}(e')| \geq P_{i^*}(T) \cdot (k-i^*) = \Omega(k).
	\end{align*}
	Now note that this happens whenever we delete an element from $T$, which by our assumption on the size of $T$, occurs with probability $\Theta(k/n)$. Lower bounding the average sensitivity by when we delete an element from $T$, our conclusion follows.
\end{proof}
We can now move onto the proof of our theorem.
\begin{proof}[Proof of Theorem~\ref{theorem:AveGeneral}]
	Again we will look at two functions, one in the large-element scheme and one in the near-equal scheme. First, take the following function:
	\begin{equation}
	f(S) = \sum_{i=1}^k x_i + \varepsilon\sum_{i=k+1}^n x_i, \text{ where } x_i = 1[e_i \in S]. \label{eqn:largeFave}
	\end{equation}
	We note that as our algorithm has a constant-factor approximation, we must select $\Omega(k)$ of the large elements with probability $\Omega(1)$ while $\Omega(k)$ steps are still remaining. Label those $\Omega(k)$ elements as the set $T$. Let the indicator variable $\mathcal{I}$ take the value $1$ when all elements in $T$ have been selected. This indicator will serve the purpose of $e_j$ in the worst-case sensitivity setting. Additionally, let $e_j$ be the element in $T$ with maximal index.
	
	Now consider the following function:
	\begin{align}
	f(S) &= \nonumber \\
	&\left((1-c\mathcal{I})\sum_{i \leq j, i \not \in \mathcal{I}} e_i\right) + \sum_{i \in \mathcal{I}}e_i + \left((1-c + \varepsilon(1-\mathcal{I}))\sum_{i=j+1}^{I_{\max}+1}e_i\right) \nonumber \\
	&+ (1-c + \varepsilon/2)\sum_{i=I_{\max}+2}^{2I_{\max}+1} e_i + \varepsilon\sum_{2I_{\max}+2}^{n} e_i \nonumber
	\end{align}
	Much like in the worst-case sensitivity scenario, due to the ordering of marginals being preserved before the indicator variable $\mathcal{I}$ takes the value $1$, we see that the two functions behave the same until that happens. This ensures that all elements in the set $T$ will be chosen in the near-equality function with probability $\Omega(1)$ while $\Omega(k)$ steps are still remaining.
	
	The analysis now follows similarly to the worst-case sensitivity setting. We get the statement for Theorem~\ref{theorem:AveGeneral} using Lemma~\ref{lem:hardmainAve}, noting we can take $T$ as defined earlier, $E_1 = \{e_1, \ldots, e_{I_{\max}+1}\}$, $E_0 = \{e_{I_{\max}+2}, \ldots, e_{2I_{\max}+1}\}$.
\end{proof}

\subsection{Distributed Algorithms}
\begin{theorem}\label{theorem:GreeDiAve}
	The Algorithm \Call{GreeDi}{} has average sensitivity $\Omega(k^2/n)$, even for bounded curvature.
\end{theorem}
\begin{proof}
	We let $\mathcal{I}$ be the indicator that takes on the value $1$ when all of $e_1, \ldots, e_{k/2}$ are selected, and take the function:
	\begin{equation}
	f(e_1,\ldots,e_n) = C\sum_{i=1}^{k/2}e_i + (1-c\mathcal{I}) \sum_{i=k/2+1}^{3k/2}e_i + (1-c/2)\sum_{i=3k/2+1}^n
	\end{equation}
	for a large constant $C$, and curvature $c$. By taking large $n$, each machine $V_i$ will have at most one of the large elements $e_1, \ldots, e_{k/2}$, with probability $1 - o_n(1)$. Further, each machine will have at most one of the elements in the second summand $e_{k/2+1}, \ldots, e_{3k/2}$ with probability $1 - o_n(1)$.
	
	We see that each machine must select the elements it contains from the first and second summand, and any remaining elements will come from the third summand.
	
	We then see that the set $S$ in the \Call{GreeDi}{} definition in Algorithm~\ref{alg:1stdistgreedy} must contain all elements from the first summand, second summand, and at least $k-1$ elements from the third summand with probability $1 - o_n(1)$.
	
	Earlier analysis of the centralized greedy algorithm now gives the result.
\end{proof}

\begin{theorem}\label{theorem:DistributedAve}
	Let $\mathcal{A}$ be a centralized sequential algorithm which satisfies the hypotheses of Theorem~\ref{mainhard} and has average sensitivity $\Omega(k^2/n)$. The corresponding distributed algorithm $\mathcal{A}_d$ also has average sensitivity $\Omega(k^2/n)$, even for bounded curvature.
\end{theorem}
\begin{proof}
	We'll consider the following function, similar to the large-element case in the centralized setting:
	\begin{equation}
	f(S) = C\sum_{i=1}^k x_i + (1-c\mathcal{I})\sum_{i=k+1}^{3k/2} + \left(1-\frac{c}{2}\right)\sum_{3k/2+1}^{n} e_i
	\end{equation}
	where the indicator $\mathcal{I}$ is the same as in the centralized case. As is the case for \Call{GreeDi}{}, considering each separate group of machines $V$, each machine in that group $V_i$ will have at most one element from the first summand, and similarly for the second.
	
	Similarly to \Call{GreeDi}{}, and to the centralized version of this proof, we see that every machine will choose the elements in the first and second summand that have been assigned to that machine, with the remainder coming from the third summand.
	
	We then see that $C_1$ will contain all elements $\{e_1, \ldots, e_{3k/2}\}$, along with at least $k-1$ elements from the third summand, and this will happen with probability $1 - o_n(1)$. Similarly to the worst-case scenario, we see that on the final run, all machines in all groups will run the greedy algorithm containing all elements in the set $\{e_1, \ldots, e_{3k/2}\}$, and at least $k-1$ other elements.
	
	Earlier analysis of the centralized algorithm now gives the result. As in the worst-case setting, we note from run $2$ onwards the algorithm must behave as in the centralized case for this function, excluding the possibility of a higher objective value, as all elements that the indicator $\mathcal{I}$ requires will always be present in the correct position.
\end{proof}

\end{document}